\documentclass[journal]{IEEEtran}

\usepackage[latin9]{inputenc}
\usepackage{color}
\usepackage{amsmath}
\usepackage{amsthm}
\usepackage{amssymb}
\usepackage{graphicx}
\usepackage[T1]{fontenc}
\usepackage[active]{srcltx}
\usepackage{color}
\usepackage{float}
\usepackage{setspace}
\usepackage{enumitem}
\usepackage[font={footnotesize,it}]{caption}
\usepackage[font={footnotesize,it}]{subcaption}
\usepackage{epsfig}
\usepackage[noadjust]{cite}
\usepackage{verbatim}
\usepackage{balance}

\makeatletter

\newtheorem{theorem}{Theorem}
\newtheorem{proposition}[theorem]{Proposition}
\newtheorem{lemma}[theorem]{Lemma}

\newcounter{algo}
\newenvironment{algo}[2]{\refstepcounter{algo}\label{#2}   \begin{center}
\begin{minipage}{0.49\textwidth}   \hrule\smallskip
\textbf{Algorithm \thealgo: #1}
\par\smallskip\hrule\smallskip\ignorespaces}{\par\smallskip\hrule
\end{minipage}
\end{center}}

\let\oldproofname=\proofname
\renewcommand{\proofname}{\rm\bf{\oldproofname}}

\renewcommand{\Re}{\mathbb{R}}

\newcommand{\Km}{\mathcal{K}_\text{{mx}}}

\newcommand{\z}{\mathbf{z}}
\newcommand{\y}{\mathbf{y}}
\newcommand{\x}{\mathbf{x}}
\newcommand{\X}{\mathbf{X}}

\newcommand{\A}{\mathbf{A}}
\newcommand{\bv}{\mathbf{b}}

\makeatother

\begin{document}

\title{\vspace{-0.5cm}Hybrid Random/Deterministic Parallel Algorithms for  Nonconvex  Big Data Optimization}

\author{Amir   Daneshmand,  Francisco Facchinei, Vyacheslav Kungurtsev, and Gesualdo Scutari\newline (the order of the authors is alphabetical$^\ast$)\vspace{-0.9cm}
\thanks{$^\ast$All the authors contributed equally to the paper.}
\thanks{A. Daneshmand and G. Scutari are with the Dept. of Electrical Engineering, at  the State Univ. of New
York at Buffalo, Buffalo, USA. Email: \texttt{<amirdane,gesualdo>@buffalo.edu}.\newline  \indent F. Facchinei is with the  Dept. of Computer, Control,
and Management Engneering, at Univ.  of Rome La Sapienza, Rome, Italy. Email: \texttt{francisco.facchinei@uniroma1.it}. \newline
\indent V. Kungurtsev is with the Agent Technology Center, Dept. of Computer Science, Faculty of Electrical Engineering, Czech Technical
University in Prague. Email: \texttt{vyacheslav.kungurtsev@fel.cvut.cz}.\newline
\indent Part of this work has been published on arxiv on June 2014.}
%
}
\maketitle
\vspace{-0.4cm}

\begin{abstract}
 We propose  a  decomposition framework for the parallel optimization
of the sum of a differentiable {(possibly nonconvex)} function and a
nonsmooth (possibly  nonseparable), convex one. The latter term is usually employed to enforce structure in the solution, typically sparsity. The main contribution of this work is a  novel \emph{parallel, hybrid random/deterministic} decomposition scheme wherein, at each iteration, a subset of (block) variables is updated at the same time by minimizing local convex approximations of the original nonconvex function. To tackle with huge-scale problems, the (block) variables to be updated are   chosen according to a \emph{mixed random and deterministic} procedure, which captures the advantages  of both pure deterministic and random update-based schemes. Almost sure convergence of the proposed scheme is established.   Numerical results show that on huge-scale problems the  proposed hybrid random/deterministic algorithm outperforms both  random and deterministic schemes.

\end{abstract}\vspace{-0.1cm}

\vspace{-0.3cm}

\section{Introduction}

\label{sec:Introduction}

We consider the minimization of the sum of a smooth (possibly \emph{nonconvex})  function $F$ and of a nonsmooth
(possibly \emph{nonseparable}) convex one $G$:
\vspace{-0.1cm}
\begin{equation}\label{eq:problem 1}
\min_{\x\in X} V(\x) \triangleq F(\x) + G(\x),
\end{equation}\vspace{-0.3cm}

\noindent
where $X$ is a closed convex set with a cartesian product structure: $X = \Pi_{i=1}^N X_i \subseteq \Re^n$.
Our focus is on problems with a huge number of variables, as those that can be encountered, for example, in
machine learning, compressed sensing, data mining, tensor factorization and completion, network optimization,
image processing, genomics, and  meteorology. We refer the reader to
\cite{tibshirani1996regression,qin2010efficient,yuan2010comparison,fountoulakis2013second,necoara2013efficient,nesterov2012gradient,
 tseng2009coordinate,beck_teboulle_jis2009,wright2009sparse,peng2013parallel,
 SlavakisGiannakis_ICASSP14,Slavakis-Giannakis-Mateos_SPMag14,DeSantisLucidiRinaldi14arxiv} and the
 books \cite{Sra-Nowozin-Wright_book11,bach2011optimization}
 as entry points to the literature.
  
Recent years have witnessed a surge of interest in these very large scale problems, and the evocative term Big Data
optimization has  been coined
to denote this new area of research.   Block Coordinate Descent (BCD) methods rapidly emerged as
a winning paradigm to attack Big Data optimization, see e.g. \cite{yuan2010comparison}. At each iteration of a BCD
method one block of variables is updated using first-order information,
while keeping all other variables fixed.  This  dramatically reduces the memory and computational requirements of
each iteration and leads to simple and scalable methods.
One of the key ingredients in a BCD method is the choice of the block of variables to update. This can be
accomplished in several ways, for example using a cyclic order
or   some greedy/opportunistic selection strategy, which aims at selecting the block leading to the largest decrease of the objective
function. The cyclic order has the advantage of being extremely simple, but the greedy strategy usually provides
faster convergence, at the cost of an increased computational effort at each iteration.
However, no matter which block selection rule is adopted, as the  dimensions of the optimization problems increase,  even  BCD methods may result inadequate. To
alleviate the ``curse of dimensionality'',  three different  kind of strategies
have been proposed, namely: (a) \emph{parallelism}, where several blocks of variables are updated simultaneously in
a multicore or distributed computing environment,
see  e.g.
\cite{bradley2011parallel,tseng2009coordinate,patriksson1998cost,FaccSagScu_ICASSP14,
FacSagScuTSP14,fercoq2014fast,fercoq2013accelerated,
LuXiao2013randomCDM,necoara2013efficient,necoara_clipici2014,nesterov2012gradient,nesterov2012efficiency,
 peng2013parallel,richtarik2012parallel,tseng2009coordinate,beck_teboulle_jis2009,wright2009sparse};
(b) \emph{random selection} of the block(s) of variables to update, see e.g.
\cite{shalev2011stochastic,nesterov2012efficiency,
richtarik2012parallel,fercoq2014fast,fercoq2013accelerated,lu2013complexity,LuXiao2013randomCDM,necoara2014random,necoara_clipici2014,
patrascu_necoara2014,richtarik2014iteration};
and (c) use of ``\emph{more-than-first-order}'' information, for example (approximated) Hessians or (parts of) the original
function itself, see e.g.
\cite{dassios2014second,FaccSagScu_ICASSP14,FacSagScuTSP14,fountoulakis2013second, yuan2012improved}.
Point  (a) is self-explanatory and rather intuitive (although the corresponding theoretical analysis is by no means
trivial); here we only remark that the vast majority of parallel BCD methods apply to \emph{convex problems only}.  Points 
(b) and (c)  need further comments.

\noindent \texttt{Point (b)}: The random selection   of variables to update (also termed \emph{random sketching})  is 
essentially as 
cheap as a cyclic selection   while   alleviating  some of the pitfalls of  cyclic rules.
 Moreover, random sketching  is relevant in  distributed environments wherein data are
not available in their entirety, but are acquired either in batches over time or over a network   (and not all nodes are equally 
responsive). In such scenarios,  one might be interested in running the optimization process at a certain instant even with the 
limited, randomly available information. The main limitation of  random selection rules  is   that they  
remain disconnected from the status of the optimization process, which  instead  is exactly the kind of behavior that 
greedy-based updates try to avoid, in favor of  faster convergence, but at the cost of more intensive computation. 


\noindent \texttt{Point (c)}: The use of ``more-than-first-order'' information also has to do with the trade-off between cost-per-iteration and overall cost of the optimization process. Using higher order or structural information may seem unreasonable,
given the huge size of the problems at hand, and in fact the accepted wisdom is that at most first-order information can be used in the Big Data environment. However,
recent studies, as those mentioned above, challenge this  wisdom and suggest that a judicious use of some kind of ``more-than-first-order'' information can lead to substantial overall improvements.


The above    pros $\&$ cons analysis suggests that it would be desirable to design a parallel algorithm for nonconvex problems
combining the benefits of   random sketching \emph{and} greedy updates, possibly using  ``more-than-first-order''
information. To the best of our knowledge, no such  algorithm exists in the literature.
In this paper,  building on our previous deterministic methods 
\cite{scutari_facchinei_et_al_tsp13,FaccSagScu_ICASSP14,FacSagScuTSP14}, we propose   a BCD-like scheme for the 
computation of stationary solutions of Problem
\eqref{eq:problem 1} filling the gap and  enjoying   \emph{all} the following features:
\vspace{-0.08cm}

\begin{enumerate}
\item It uses a random selection rule for the blocks, followed by  a deterministic subselection;
\item It can classically tackle  separable convex function $G$, i.e.,  $G(\x) = \sum_i G_i(\x_i)$, but   also  
\emph{nonseparable} functions $G$;
\item It can deal with a nonconvex functions $F$;
\item It can use both first-order and higher-order information;
\item It is parallel;
\item It can use inexact updates;
\item It converges {\em almost surely}, i.e. our convergence results are  of the form ``with probability one''.
\end{enumerate}

\noindent
As far as we are aware of, this is the first algorithm enjoying all these properties, 
even in the convex case.
The combination  of all the features 1-7 in one single algorithm is  a major achievement in itself, which  offers great flexibility to develop tailored instances of solutions methods within the same framework (and thus all converging under the same unified conditions). 
Last but not least, our experiments show  impressive performance of the proposed methods, outperforming state-of-the-art solution scheme (cf. Sec.  \ref{sec:Num}). 
As a final remark, we underline that, at more methodological level,   the combination of all features 1-7 and, in particular, the need to conciliate random and deterministic strategies,  led   to the development of  a new type of convergence analysis 
(see Appendix \ref{App:random_properties}) which is also of interest {\em per se} and could bring to further   developments.

Below we  further comment on some of features 1-7, compare to existing results, and detail 
our  contributions. 

\noindent \texttt{Feature 1}: As far as we are aware of, the idea of making a random selection and then perform a greedy subselection has been 
previously discussed only in \cite{NIPS2012_4674}. However,  results therein i) are only for 
\emph{convex} problems with a \emph{specific} structure; ii) are based on a regularized first-order model; iii)  require a very stringent 
``spectral-radius-type'' condition, 
which   severely limits the degree   of parallelism$-$the 
maximum  number of variables that can be simultaneously updated at each iteration while guaranteeing  convergence; and iv)  convergence results are in terms of  expected value of the objective function. The proposed algorithmic framework  expands vastly on this 
setting, while enjoying also all  properties 2-7. In  particular,  it is the first hybrid random/greedy scheme for \emph{nonconvex nonseparable} functions, and it allows  \emph{any} degree of parallelism (i.e., the update  of any number of variables); and  all this is achieved under much weaker convergence conditions 
than those in  \cite{NIPS2012_4674}, satisfied by most of practical problems.
 Numerical results  show that the proposed hybrid schemes   updating   greedily 
just some blocks within the  pool of those   selected by  a random rule 
is very effective, 
and seems to preserve 
the advantages of both  random and deterministic selection rules.

\noindent \texttt{Feature 2}: The ability of dealing with some classes of nonseparable convex functions  
has been documented in  \cite{auslender1976optimisation, tseng2001convergence,razaviyayn2013unified}, \emph{but only for deterministic and 
  sequential schemes}; our approach extends also to    \emph{parallel},  \emph{random} schemes. 
 
 \noindent \texttt{Feature 3}: 
  The list    of works dealing with BCD methods for nonconvex $F$'s   is  short:
\cite{patrascu_necoara2014, LuXiao2013randomCDM} for \emph{random sequential} methods; and    
\cite{tseng2009coordinate,patriksson1998cost,FaccSagScu_ICASSP14,FacSagScuTSP14,razaviyayn2014parallel}
for \emph{deterministic parallel} ones.  The only (very recent) paper dealing with \emph{random parallel} methods for 
  nonconvex $F$'s  is the arxiv submission  \cite{razaviyayn2014parallel}, which  however does not enjoy the key
properties  1, 2, and 6.

 \noindent \texttt{Feature 4}: 
 We want to stress  the ability of the proposed algorithm to exploit in a systematic way  ``more-than-first-order''
information. At each iteration of a BCD
method, one block of variables is updated using a (possibly regularized) first-order model of the objective function,
while keeping all other variables fixed. Our method, following the approach first explored in 
\cite{scutari_facchinei_et_al_tsp13,FaccSagScu_ICASSP14,FacSagScuTSP14}  provides the flexibility  of using more sophisticated 
models. For example, i) one could use a Newton-like approximation;  or ii)  suppose that in (\ref{eq:problem 1}) $F= F_1 + F_2$, where
$F_1$ is convex and $F_2$ is not. Then, at iteration $k$,  one could base the update of the $i$-th block on the approximant 
 $F_1(\mathbf{x}_i, \mathbf{x}_{-i}^k) +   \nabla_{\x_i} F_2 (\x^k)^T (\x_i - \x_i^k)   + G(\mathbf{x}_i,\x_{-i}^k)$,
where $\x_{-i}$ denotes the vector obtained from $\mathbf{x}$ by deleting   $\x_i$. The logic here is that instead 
of linearizing the whole function $F$ we only linearize the difficult, nonconvex part $F_2$.
In this light  we can  also better appreciate the importance of feature 6, since if we go for more complex approximants, the 
ability to deal with inexact solutions becomes important. 

\noindent \texttt{Feature 6}: 
Inexact solution methods have been little studied.
Papers \cite{goodman2008exponential,chang2008coordinate,yuan2010comparison} (somewhat indirectly) consider 
some of these issues in the specialized context of $\ell_2$-loss linear support vector machines. A more systematic treatment of 
inexactness of the solution of a first-order model  is documented in \cite{tappenden2013inexact}, in the context of random 
sequential BCD methods for \emph{convex} problems. Our results in this paper are based on our previous works 
\cite{scutari_facchinei_et_al_tsp13,FaccSagScu_ICASSP14,FacSagScuTSP14}, where both the use of ``more-than-first-order'' 
models and inexactness are introduced and rigorously analyzed   in the  context of parallel, deterministic methods.  
This paper  extends results  in \cite{scutari_facchinei_et_al_tsp13,FaccSagScu_ICASSP14,FacSagScuTSP14}   to  the random, parallel schemes for   \emph{nonconvex} objective functions, and constitute the 
first study of these issues in this setting.


As a final remark, we observe that a large portion of works mentioned so far are   interested in (global) complexity analysis. 
Of course this is an 
important  topic, but it is outside the scope of this paper. Note that, with the exception of 
\cite{patrascu_necoara2014}, all papers dealing with complexity analyses, study (regularized) \emph{gradient-type}  
methods for \emph{convex} problems. 
Given our expanded setting, we believe it is more fruitful  to concentrate on proving convergence  and verifying the   practical 
effectiveness  of our algorithms.

The paper is organized as follows. Section \ref{sec:Pro} formally introduces the optimization  problem  along with the main assumptions  under which it is studied and also discusses some  technical points.
 The proposed  algorithmic framework  and its convergence properties are introduced in Section \ref{sec:Main Results},
 while  numerical results are presented in Section  \ref{sec:Num}. Section \ref{sec:conclusions} draws some conclusions. All proofs  are given in the Appendix.\vspace{-0.2cm}



\section{Problem Definition and Preliminaries}\label{sec:Pro}

We consider Problem \eqref{eq:problem 1},
where the feasible set $X=X_1\times \cdots \times  X_N$ is a Cartesian product of lower dimensional convex sets $X_i\subseteq \Re^{n_i}$, and $\mathbf x\in \Re^n$ is partitioned accordingly: $\mathbf x = (\mathbf  x_1, \ldots, \mathbf x_N)$, with each $\mathbf  x_i \in \Re^{n_i}$; we denote by $\mathcal N\triangleq \{1,\ldots ,N\}$ the set of the $N$ blocks. The function $F$ is  smooth (and not necessarily convex and separable)
and $G$  is convex, and  possibly nondifferentiable and nonseparable.  Some widely-used  choices for $G(\x)$ are $c\|\x\|_1$ and  $c \sum_{i=1}^N\|\x_i\|_2$, from which one can see that
Problem \eqref{eq:problem 1} includes many popular Big Data optimization problems, such as Lasso, group Lasso, sparse logistic regression,  $\ell_2$-loss Support Vector Machine,
 Nuclear Norm Minimization, and Nonnegative Matrix (or Tensor) Factorization problems.

\noindent \textbf{Assumptions.} Given \eqref{eq:problem 1}, we make the following  blanket assumptions:

\begin{description}[topsep=-2.0pt,itemsep=-2.0pt]
\item[\rm  (A1)]  Each $X_i$ is nonempty, closed, and convex;\smallskip
\item[\rm  (A2)] $F$ is $C^1$ on an open set containing $X$;\smallskip
\item[\rm  (A3)]  $\nabla F$ is   Lipschitz continuous
on $X$ with constant $L_{F}$;\smallskip
\item[\rm  (A4)] $G$ is continuous and convex on $X$ (possibly nondifferentiable and nonseparable);\smallskip
\item[\rm  (A5)] $V$ is coercive.
\end{description}
\medskip

\noindent
Note that the above assumptions are standard and are satisfied by most of the problems of practical interest. For instance,  A3 holds automatically if $X$ is bounded, whereas  A5 guarantees the existence of a solution.


%


With the advances of multi-core architectures, it is desirable to develop  {\em parallel} solution methods for Problem \eqref{eq:problem 1} whereby operations can be carried out on some or  (possibly) all   (block) variables $\x_i$ at
the \emph{same} time. The most natural   parallel (Jacobi-type) method one can think of is  updating  \emph{all} blocks simultaneously: given $\mathbf{x}^k$, each (block) variable $\mathbf{x}_i$  is updated   by solving the following subproblem
\begin{equation}\label{eq:plain_Jac}
\mathbf{x}_i^{k+1}\in  \underset{\mathbf{x}_i\in X_i}{\text{argmin}}\,\, \left\{ F(\mathbf{x}_i, \mathbf{x}_{-i}^k) + G(\mathbf{x}_i,\x_{-i}^k)\right\}.
\end{equation}
Unfortunately this method converges only under very restrictive conditions  \cite{Bertsekas_Book-Parallel-Comp} that are seldom verified in practice (even in the absence of the nonsmooth part $G$).  Furthermore, the exact computation of $\mathbf{x}_i^{k+1}$ may be difficult and computationally  too expensive.

To cope with these issues, a natural approach is to replace the (nonconvex) function $F(\bullet, x^k_{-i})$ by a suitably chosen local convex approximation $\widetilde{F}_i(\x_i;\x^k)$, and solve instead  the convex problems (one for each block)
\begin{equation}\label{eq:plain_Jac bis}
\mathbf{x}_i^{k+1}\in  \underset{\mathbf{x}_i\in X_i}{\text{argmin}}\,\, \left\{ \tilde{h}_{i}(\mathbf{x}_{i};\x^k)\triangleq \widetilde{F}_i(\x_i;\x^k) + G(\x_i;\x_{-i}^k)\right\},
\end{equation}
with the understanding that the minimization in \eqref{eq:plain_Jac bis} is simpler than that in \eqref{eq:plain_Jac}. Note that the function $G$ has not been touched; this is
  because i) it is generally much more difficult to find a ``good'' approximation of a nondifferentiable function than of a differentiable one; ii) $G$ is already convex; and iii) the functions $G$ encountered in practice do not make the optimization problem  \eqref{eq:plain_Jac bis}  difficult (a closed form solution is available for a large classes of $G$'s, if $\widetilde{F}_i(\x_i;\x^k)$ are properly chosen). In this work we assume that the approximation functions   $\widetilde{F}_i(\mathbf{z};\mathbf{w}) :  X_i \times  X \to \Re$,     have the following properties (we denote by $\nabla \widetilde{F}_i$
the partial gradient of $\widetilde{F}_i$ with respect to the first argument   $\mathbf{z}$):\smallskip
\begin{description}
\item[\rm  (F1)]
 $\widetilde{F}_{i} (\mathbf{\bullet}; \mathbf{w})$ is uniformly strongly  convex with constant $q>0$ on $ X_i$;
\item[\rm  (F2)]  $\nabla \widetilde{F}_{i} (\mathbf{x}_i;\mathbf{x}) = \nabla_{{\mathbf{x}}_i} F(\mathbf{x})$ for all $\mathbf{x} \in  X$;
\item[\rm  (F3)]  $\nabla \widetilde{F}_{i} (\mathbf{z};\mathbf{\bullet})$ is Lipschitz continuous
on $ X$ 
 for all $\mathbf{z} \in  X_i$.
\end{description}
\smallskip

Such a  function  $\widetilde{F}_i$ should be regarded as a (simple) convex approximation of $F$ at the point $\mathbf{x}$ with respect to the block of variables $\mathbf{x}_i$ that
preserves the first order properties of $F$ with respect to  $\mathbf{x}_i$. Note that, contrary to most of the works in the literature (e.g., \cite{razaviyayn2013unified}), we do not require $\widetilde{F}_i$ to be a global  {\em upper} approximation of  $F$, which significantly enlarges the range of applicability of the proposed solution methods.

  The most popular choice for
$\widetilde{F}_i$ satisfying F1-F3  is
\begin{equation}\label{gradient-prox}
\widetilde{F}_i(\x_i;\x^k) = F(\x^k) + \nabla_{\x_i} F(\x^k)^T(\x_i-\x_i^k) + \frac{\tau_i}{2} \| \x_i-\x_i^k\|^2,
\end{equation}
with $\tau_i>0$. This is essentially the way   a new iteration is computed in most  (block-)BCDs for the
solution of (group) LASSO problems and its generalizations. When $G\equiv 0$, this choice gives rise to a gradient-type scheme; in fact we obtain $\mathbf{x}_i^{k+1}$ simply by a shift along the antigradient. As we discussed in the introduction, this is a first-order method, so it seems advisable, at least in some situations, to use more informative $\widetilde{F}_i$-s. If     $F(\x_i, \x_{-i}^k)$ is convex, an alternative is to take  $\widetilde{F}_i(\x_i;\x^k)$ as a second order approximation of $F(\x_i, \x_{-i}^k)$, i.e.,
\begin{equation}
\begin{array}{ll}
\widetilde{F}_i(\x_i; \x^k) = F(\x^k) + \nabla_{\x_i}F(\x^k)^T(\x_i - \x_i^k) \smallskip\\ \hspace{2cm}+ \frac{1}{2} (\x_i - \x_i^k)^T \left( \nabla^2_{\x_i\x_i} F(\x^k)+ qI\right)  (\x_i - \x_i^k),\end{array}
\end{equation}
where $q$ is nonnegative and can be taken to be zero if  $F(\x_i, \x_{-i}^k)$ is actually strongly convex.
When $G\equiv 0$, this essentially corresponds to taking a Newton step in minimizing the ``reduced'' problem $\min_{\x_i\in X_i}F(\x_i, \x_{-i}^k)$. Still in the case of convex  $F(\x_i, \x_{-i}^k)$, one could also take just $$\widetilde{F}_i(\x_i; \x^k) = F(\x_i, \x_{-i}^k),$$ which preserves the whole structure of the function.   Other valuable choices tailored to specific applications are discussed in \cite{scutari_facchinei_et_al_tsp13,FacSagScuTSP14}.  As a guideline, note that our method, as we shall describe in details shortly, is based on the iterative (approximate) solution of problem \eqref{eq:plain_Jac bis} and therefore a balance should be aimed at between the accuracy of the approximation $\tilde F$ and the ease of solution of   \eqref{eq:plain_Jac bis}.  Needless to say, the option (\ref{gradient-prox})
 is the less informative one, although it usually makes the computation of the solution of \eqref{eq:plain_Jac bis} a cheap task.

\noindent \textbf{Best-response map:} Associated with each $i$ and point $ \mathbf{x}^k \in X$, under F1-F3,  we can define the following optimal block solution map:\vspace{-0.1cm}
\begin{equation}\label{eq:decoupled_problem_i}
\widehat{\mathbf{x}}_{i}(\mathbf{x}^k)\triangleq\underset{\mathbf{x}_{i}\in X_{i}}{\mbox{argmin}\,}\tilde{h}_{i}(\mathbf{x}_{i};\mathbf{x}^{k}).\vspace{-0.1cm}
\end{equation}
 Note that $\widehat{\mathbf{x}}_{i}(\mathbf{x}^{k})$
is always well-defined, since the optimization problem in (\ref{eq:decoupled_problem_i})
is strongly convex. Given (\ref{eq:decoupled_problem_i}),
we can then introduce the solution map
\begin{equation}\label{best-response}
X\ni\mathbf{y}\mapsto\widehat{\x}(\mathbf{y})\triangleq\left(\widehat{\x}_{i}(\y)\right)_{i=1}^{N}.
\end{equation}

Our algorithmic framework is based on solving in parallel a suitable selection of  subproblems \eqref{eq:decoupled_problem_i}, converging thus to  {\em fixed-points} of
$\widehat{\x}(\bullet)$ (of course the selection varies at each iteration). It is then natural to ask which relation exists  between these fixed points and the stationary solutions of Problem \eqref{eq:problem 1}. To answer this key question, we recall first a few definitions.

\begin{description}
\item[\textbf{Stationarity:}]
A point $\x^*$  is a stationary point of \eqref{eq:problem 1} if   a subgradient $\boldsymbol{\xi} \in \partial G(\x^*)$ exists such
that
$(\nabla F(\x^*) $ $+ \boldsymbol{\xi} )^T(\y-\x^*) \geq 0$ for all $\y\in X$.  \medskip
\item[\textbf{Coordinate-wise stationarity:}] A point $\x^*$  is a coordinate-wise stationary point
of \eqref{eq:problem 1} if   subgradients $\boldsymbol{\xi}_i \in \partial_{\boldsymbol{\xi}_i} G(\x^*),$ with $i \in \mathcal N$, exist such that
$(\nabla_{\x_i} F(\x^*)  + \boldsymbol{\xi}_i)^T(\y_i-\x^*_i) \geq 0$, for all $\y_i\in X_i$ and $i\in \mathcal N$. 
\smallskip
\end{description}
Of course,  if $F$ is convex, stationary points
coincide with its global minimizers. In words, a coordinate-wise stationary solution is  a point for which $\x^*$ is stationary w.r.t. every block of
variables. It is clear that a stationary point is always a coordinate-wise stationary point;  the converse however  is not always true, unless extra conditions on $G$ are satisfied.
\begin{description}
\item[\textbf{Regularity:}]  Problem \eqref{eq:problem 1} is {\em regular} at a coordinate-wise stationary point $\x^*$ if
$\x^*$ is also a stationary point of the problem.\smallskip\end{description}

Regularity at $\x^*$ is a rather weak requirement,
and is easily seen to be implied, in particular, by the
following two conditions: 
\begin{description} \item[(a)] $G$ is separable (still nonsmooth), i.e., $G(\x) = \sum_i G_i(\x_i)$;\smallskip
\item[(b)] $G$ is continuously differentiable around $\x^*$.\end{description} Note that (a) is assumed in practically all papers dealing with deterministic/random BCD methods (with the exception of \cite{tseng2001convergence,razaviyayn2013unified}, where however only \emph{sequential} schemes are proposed). Regularity can well occur also for nonseparable functions. For instance,  consider the function arising in logistic regression problems $F(\x) = \sum_{j=1}^m \log(1 + e^{-a_ij\y_j^T\x }),$
with $X=\Re^n$, and $\y_j\in \Re^n$ and $a_j\in \{-1, 1\}$ being  given
constants. Now, choose $G(\x) = c\|\x\|_2$;
the resulting function is continuously differentiable, and therefore regular, at any stationary point but $\x^*\neq \mathbf{0}$.
It is easy to verify that $V$ is also regular at $\x=\mathbf{0}$, provided that $c< \log 2$.

The following proposition is elementary and elucidates the connections between stationarity conditions of Problem \eqref{eq:problem 1}
and fixed-points of $\widehat{\x}(\bullet)$.\vspace{-0.1cm}
\begin{proposition}\label{prop_fixed-points}
Given Problem  \eqref{eq:problem 1} under A1-A5 and  F1-F3, the following hold:
\begin{description}
\item[i)] The set of fixed-points of $\widehat{\x}(\mathbf{\bullet})$ coincides with the coordinate-wise stationary points of
Problem \eqref{eq:problem 1};

\item[ii)] If, in addition,  Problem \eqref{eq:problem 1} is regular at a fixed-point of $\widehat{\x}(\bullet)$, then such a  fixed-point is also  a stationary point of the problem.
\end{description}
\end{proposition}\vspace{-0.1cm}

Other properties of the best-response map $\widehat{\x}(\mathbf{\bullet})$  that are instrumental to prove convergence of the proposed algorithm are introduced in Appendix \ref{sec:BR_properties}.\vspace{-0.2cm}

\section{Algorithmic Framework\label{sec:Main Results}}

We are   ready to describe our  algorithmic framework. We begin introducing a formal  description  of its salient characteristic,  the novel hybrid  random/greedy  block selection rule. 

The   random block selection works as follows: at each iteration $k$, a random set $\mathcal{S}^k\subseteq \mathcal N$ is generated, and the blocks $i\in \mathcal{S}^k$ are  the potential candidate variables  to    update  in parallel. The set $\mathcal{S}^k$ is a realization of a random set-valued mapping  $\boldsymbol{\mathcal{S}}^k$ with
values in the power set of $\mathcal N$. To keep the proposed scheme as general as possible, we do not constraint $\boldsymbol{\mathcal{S}}^k$ to any specific distribution; we only require   that, at each iteration $k$, each block $i$ has a chance (positive probability, possibly nonuniform) to be
selected.
\begin{description}
\item[\rm  (A6)] The   sets  $\mathcal{S}^k$ are realizations of  independent random set-valued  mappings $\boldsymbol{\mathcal{S}}^k$   such that  $\mathbb{{P}}(i\in \boldsymbol{\mathcal{S}}^k)\geq p$, for all $i=1,\ldots ,N$ and $k\in \mathbb{N}_+$, and some $p>0$.
\end{description}

A random selection rule   $\boldsymbol{\mathcal{S}}^k$ satisfying A6 will be called \emph{proper sampling}. Several proper sampling rules will be discussed in details shortly.

\smallskip

As already discussed in the introduction, the random selection of blocks seems becoming beneficial when the dimensions of the problem increase significantly. But  recent results in \cite{FacSagScuTSP14,peng2013parallel,li2009coordinate,NIPS2011_4425} strongly suggest that a greedy approach updating only the  ``promising'' blocks  is an important ingredient of an efficient algorithm. Of course,  for very large scale problems, checking   whether a block is promising or not might become  computationally demanding  and   thus time consuming. To avoid this burden while capturing the benefits of both strategies, the proposed approach consists in  combining  random and greedy updates in the following form. First,  a random selection is performed$-$the set $\mathcal S^k$ is generated. Second,  a greedy procedure is run to select  {\em in  the pool} $\mathcal S^k$  only the subset of  blocks, say $\hat{\mathcal{S}}^k$, that are ``promising'' (according to a prescribed criterion). Finally all the blocks in $\hat{\mathcal{S}}^k$ are updated  in parallel.
To complete the description of such an hybrid random/greedy selection,  the notion of ``promising'' block needs to be made formal, which is done next.

Since  $\x^k_i$ is an optimal solution of \eqref{eq:decoupled_problem_i} if and only if
$\widehat{\x}_{i}(\x^k)=\x^k_i$, a natural  distance of  $\x^k_i$ from the optimality is  $d_i^{\,k}\triangleq \|\widehat{\x}_{i}(\x^k)-\x^k_i\|$. The blocks  in ${\cal S}^k$ to be updated can be then chosen based on    such an  optimality measure (e.g., opting for blocks exhibiting larger $d_i^{\,k}$'s). However, this choice requires the  computation of  the solutions $\widehat{\x}_{i}(\x^k)$, for all   $i \in {\cal S}^k$, which in some applications might be  still computationally  too expensive.
Building on the same idea, we can introduce alternative, less expensive metrics by replacing the distance   $\|\widehat{\x}_{i}(\x^k)-\x^k_i\|$ with a computationally cheaper {\em error bound},
i.e., a function $E_i(\x)$ such that
\begin{equation}\label{eq:error bound}
\underbar s_i\|\widehat{\x}_{i}(\x^k)-\x^k_i\| \le E_i(\x^k) \leq \bar s_i\|\widehat{\x}_{i}(\x^k)-\x^k_i\|,
\end{equation}
for some $0< \underbar s_i \le \bar s_i$. Of course one can always set  $ E_i(\x^k) = \|\widehat{\x}_{i}(\x^k)-\x^k_i\| $, but other choices are also possible, we refer the interested reader to \cite{FacSagScuTSP14} for more details.

The proposed hybrid random/greedy scheme capturing all the features 1)-6) discussed in Sec. I is formally given in Algorithm 1.   Note that in step S.3 inexact calculations of $\widehat{\mathbf{x}}_i$ are allowed, which is another noticeable and useful feature: one can reduce the cost per iteration without affecting too much, experience shows,
 the empirical convergence speed. 
 In step S.5 we introduced a memory in the variable updates:  the new point $\x^{k+1}$ is a convex combination via $\gamma^k$ of  $\x^k$ and $\widehat{\mathbf{z}}^{k}$.   The step-size $\gamma^k$ plays a key rule in the convergence, and needs to be properly tuned, as  specified  in Theorem \ref{Theorem_convergence_inexact_Jacobi}, which summarizes the convergence properties of Algorithm \ref{alg:general}.\begin{algo}{\textbf{Hybrid Random/Deterministic  Flexible Parallel Algorithm (HyFLEXA)}} S$\textbf{Data}:$ $\{\varepsilon_{i}^{k}\}$
for $i\in\mathcal{N}$, $\boldsymbol{{\tau}}\geq\mathbf{0}$, $\{\gamma^{k}\}>0$,
$\mathbf{x}^{0}\in X$, $\rho \in (0,1]$.

\hspace{1cm} Set $k=0$.

\texttt{$\mbox{(S.1)}:$}$\,\,$If $\mathbf{x}^{k}$ satisfies a termination
criterion: STOP;

\texttt{$\mbox{(S.2)}:$} Randomly generate a set of blocks $\mathcal{S}^k\subseteq \{1,\ldots,N\}$

\texttt{$\mbox{(S.3)}:$}  Set $M^k \triangleq  \max_{i\in \mathcal{S}^k} \{E_i(\x^k)\}$.

\hspace{1.26cm} Choose a  subset  $\hat{\mathcal{S}}^k\subseteq \mathcal{S}^k$  that  contains at least

\hspace{1.24cm} one index $i$ for which
$E_i(\x^k) \geq \rho M^k.$

\texttt{$\mbox{(S.4)}:$} For all $i\in\hat{\mathcal{S}}^k$, solve (\ref{eq:decoupled_problem_i})
with accuracy $\varepsilon_{i}^{k}:$

\hspace{1.96cm} find $\mathbf{z}_{i}^{k}\in X_i$ s.t. $\|\mathbf{z}_{i}^{k}-\widehat{\mathbf{x}}_i\left(\mathbf{x}^{k}\right)\|\leq\varepsilon_{i}^{k}$;

\hspace{1.24cm} Set $\widehat{\z}^k_i =  \mathbf{\z}_{i}^k$ for $i\in \hat{\mathcal{S}}^k$ and
$\widehat{\z}^k_i = \x^k_i$ for $i\not \in \hat{\mathcal{S}}^k$

\texttt{$\mbox{(S.5)}:$} Set $\mathbf{x}^{k+1}\triangleq\mathbf{x}^k+\gamma^{k}\,(\widehat{\mathbf{z}}^{k}-\mathbf{x}^{k})$;

\texttt{$\mbox{(S.6)}:$} $k\leftarrow k+1$, and go to \texttt{$\mbox{(S.1)}.$}
\label{alg:general}
 \end{algo}

\begin{theorem} \label{Theorem_convergence_inexact_Jacobi}Let
$\{\mathbf{x}^{k}\}$ be the sequence generated by
Algorithm \ref{alg:general}, under A1-A6.
 Suppose  that $\{\gamma^{k}\}$
and $\{\varepsilon_{i}^{k}\}$ satisfy the following conditions: i)
$\gamma^{k}\in(0,1]$; ii) $\gamma^{k}\rightarrow0$; iii) $\sum_{k}\gamma^{k}=+\infty$;
iv) $\sum_{k}\left(\gamma^{k}\right)^{2}<+\infty$;
and v)  $\varepsilon_{i}^{k}\leq \gamma^k  \alpha_1\min\{\alpha_2, 1/\|\nabla_{\mathbf{x}_i} F(\mathbf{x}^k)\| \}$
for all $i\in {\cal N}$ and  some nonnegative  constants $\alpha_1$ and $\alpha_2$.
Additionally, if inexact solutions are used in Step 3, i.e.,  $\varepsilon_{i}^{k}>0$ for some $i$ and infinite $k$, then
assume also that $G$ is globally Lipschitz on $X$.
Then, either Algorithm \ref{alg:general} converges in a finite number of iterations to a fixed-point of $\hat{\mathbf{x}}(\bullet)$
of \eqref{eq:problem 1} or there exists at least one limit point of
  $\{\mathbf{x}^{k}\}$ that is   a fixed-point of $\hat{\mathbf{x}}(\bullet)$ w.p.1.
\end{theorem}
\begin{proof} See Appendix \ref{proof_Th1}.\end{proof}\vspace{-0.1cm}

The convergence results in Theorem \ref{Theorem_convergence_inexact_Jacobi} can be strengthened  when $G$ is  separable.

\begin{theorem} \label{Theorem_G_separable} In the setting of Theorem \ref{Theorem_convergence_inexact_Jacobi}, suppose in addition that $G(\x)$ is separable, i.e., $G(\x)=\sum_{i\in {\mathcal N}} G_i(\x_i)$. Then,  either Algorithm \ref{alg:general} converges in a finite number of iterations to a stationary solution of Problem
 \eqref{eq:problem 1} or \emph{every} limit point of
  $\{\mathbf{x}^{k}\}$ is      a stationary solution of Problem
 \eqref{eq:problem 1}  w.p.1.
\end{theorem}
\begin{proof} See Appendix \ref{proof_Th2}.\end{proof}
 
\noindent \textbf{On the random choice of ${\cal S}^k$.}  We discuss next some proper sampling rules $
\boldsymbol{\mathcal{S}}^k$ that can be used in   Step 3 of the algorithm to generate the random sets  $\mathcal S^k$;
for notational simplicity the iteration index $k$ will be omitted.  The sampling rule $\boldsymbol{\mathcal{S}}$ is uniquely
characterized by the probability mass function $$\mathbb{{P}}({\mathcal{S}})\triangleq\mathbb{{P}}
\left(\boldsymbol{\mathcal{S}}=\mathcal S\right),\,\,\mathcal S\subseteq \mathcal N,
 $$
 which assign probabilities to the subsets $\mathcal{S}$ of $\mathcal N$. Associated with $\boldsymbol{\mathcal{S}}$, define the
 probabilities $q_j\triangleq \mathbb{{P}}(\left |\boldsymbol{\mathcal{S}}\right|=j)$, for $j=1,\ldots ,N$. The following
 proper   sampling rules, proposed in \cite{richtarik2012parallel} for convex problems with separable $G$,
 are instances of rules satisfying A6, and are used in our computational experiments.

\noindent $-$ \emph{Uniform (U) sampling.} All blocks get selected with the same (non zero) probability:
\[\mathbb{{P}}(i\in \boldsymbol{\mathcal{S}})=\mathbb{{P}}(j\in \boldsymbol{\mathcal{S}})=\dfrac{\mathbb{E}\left[ \left|
\boldsymbol{\mathcal{S}} \right|\right]}{N},\quad \forall i\neq j\in \mathcal N.\]

 \noindent $-$ \emph{Doubly Uniform (DU) sampling}.  All sets $\mathcal{S}$ of equal cardinality are generated with equal
 probability,   i.e., $\mathbb{P}(\mathcal{S})=\mathbb{P}(\mathcal{S}^{'})$, for all  $\mathcal{S}, \mathcal{S}^{'}\subseteq
 \mathcal N$ such that $|\mathcal S|=|\mathcal S^{'}|$. The density function is then $$\mathbb{{P}}
 ({\mathcal{S}})=\dfrac{{q_{\left|\mathcal{S}\right|}}}{\left(\begin{array}{c}
n\\
\left|\mathcal{S}\right|
\end{array}\right)}.$$
\noindent $-$ \emph{Nonoverlapping Uniform (NU) sampling}. It is a uniform sampling rule assigning positive probabilities only to sets forming a partition of $\mathcal N$. Let $\mathcal{S}^1,\ldots ,\mathcal{S}^P$ be a partition of $\mathcal{N}$,  with each $\left|\mathcal{S}^{i}\right|>0$, the density function of the NU sampling is:
 $$\mathbb{{P}}(\mathcal{S})=\left\{ \begin{array}{ll}
\dfrac{{1}}{P},\quad & \mbox{if }\mathcal{S}\in\left\{ \mathcal{S}^{1},\ldots,\mathcal{S}^{P}\right\}\smallskip \\
0 & \mbox{otherwise}
\end{array}\right.$$
 which corresponds to $\mathbb{{P}}
 (i\in {\mathcal{S}})=N/P$, for all $i\in \mathcal N$.

 A special case of the DU sampling  that we found very effective in our   experiments is the so called ``nice sampling''.

\noindent $-$ \emph{Nice Sampling (NS)}. Given an integer $0\leq \tau \leq N$, a $\tau$-nice sampling is a DU sampling with $q_{\tau}=1$ (i.e., each subset of $\tau$ blocks is chosen with the same probability).

The NS allows us to control the degree of parallelism of the algorithm by tuning   the cardinality $\tau$ of the   random sets generated at each iteration, which makes this rule  particularly appealing in a multi-core environment. Indeed, one can set $\tau$  equal to the  number of available cores/processors, and assign each block coming out from the greedy selection (if implemented) to a dedicated processor/core.

As a final remark, note that the DU/NU   rules contain as special cases fully parallel and sequential updates, wherein at each iteration a \emph{single} block is updated uniformly at random, or \emph{all} blocks are updated.

\noindent $-$ \emph{Sequential sampling}: It is a DU sampling with $q_1=1$, or a NU sampling with $P=N$ and $\mathcal S^j={j}$, for $j=1,\ldots,P$.

 \noindent $-$ \emph{Fully parallel  sampling}: It is a DU sampling with $q_N=1$, or a NU sampling with $P=1$ and $\mathcal S^1=\mathcal N$.

 Other interesting uniform and nonuniform practical rules (still satisfying A6) can be found in    \cite{richtarik2012parallel, richtarik2013optimal}, to which we refer the interested reader for further details.. \smallskip

\noindent \textbf{On the choice of the step-size $\gamma^k$.}
An example of step-size rule satisfying   Theorem \ref{Theorem_convergence_inexact_Jacobi}i)-iv)  is:
given $0 < \gamma^{0} \le 1$, let\vspace{-0.1cm}
\begin{equation}\label{eq:gamma}
\gamma^{k}=\gamma^{k-1}\left(1-\theta\,\gamma^{k-1}\right),\quad k=1,\ldots,\vspace{-0.1cm}
\end{equation}
where $\theta\in(0,1)$ is a given constant. Numerical results in Section \ref{sec:Num} show  the effectiveness of   \eqref{eq:gamma} on specific problems.
We remark that it is   possible to
prove convergence of Algorithm 1 also using other step-size rules, including  a standard Armijo-like line-search procedure or a (suitably small) constant step-size. Note that differently from most of the schemes in the literature, the tuning of the step-size  does not require the knowledge of the problem parameters (e.g.,  the Lipschitz constants of $\nabla F$ and $G$).
\vspace{-0.3cm}

\section{Numerical Results}\label{sec:Num}\vspace{-0.1cm}

In this section we present some  preliminary experiments providing a solid evidence of the viability of our approach; they clearly show  that our  framework  leads to practical methods that exploit well parallelism and compare favorably to existing schemes, both deterministic and random.

Because of space limitation, we present results only for    (synthetic) LASSO problems, one of the most   studied instances of (the convex version of) Problem (\ref{eq:problem 1}), corresponding to $F(\x) = \|\A\x -\bv\|^2 $,  $G(\x)= c \|\x\|_1$, and $X=\Re^n$. Extensive experiments on more varied (nonconvex) classes of Problem (\ref{eq:problem 1}) are the subject of a separate work. 

All codes have been written in C++ and use the Message Passing Interface for parallel operations. All algebra is performed by using the Intel Math Kernel Library (MKL). The
algorithms were tested on the General Compute Cluster of the
Center for Computational Research at the SUNY Buffalo.   In particular for our experiments we used a partition composed of 372 DELL 32x2.13GHz Intel
E7-4830  Xeon Processor  nodes with
512 GB of DDR4 main memory and QDR InfiniBand 40Gb/s
network card.


\noindent \textbf{Tuning of Algorithm 1}:  The most successful class of random and deterministic methods for LASSO problem are (proximal) gradient-like schemes, based on a  linearization of  $F$. As a major departure from current schemes,  here we propose  to   better exploit the structure  of $F$ and  use  in Algorithm 1 the following best-response: given a scalar partition of the variables (i.e., $n_i=1$ for all $i$), let 
\begin{equation}\label{eq:proposal 2}
 \widehat x_i(\x^k) \triangleq \underset{x_{i}\in \mathbb{R}}{\mbox{argmin}\,} \left\{ F(x_i,  \x^k_{-i}) + \frac{\tau_i}{2}
( x_i- x_i^k)^2 + \lambda|x_i|\right\}.
\end{equation}
Note that $\widehat x_i(\x^k)$ has a closed form expression (using a soft-thresholding operator \cite{beck_teboulle_jis2009}).

The free
parameters of Algorithm 1 are chosen as follows. The proximal gains  $\tau_i$ and the step-size $\gamma$ are tuned as in \cite[Sec. VI.A]{FacSagScuTSP14}.
 The error bound function  is chosen as  $E_i(\x^k) =\|\widehat \x_i (\x^k) - \x_i^k\|$, and, for any realization    $\mathcal S^k$,  the subsets   $\hat{\mathcal S}^k$  in   S.3 of the algorithm  are chosen as
\begin{equation}\label{S^k_simulations}
\hat{\mathcal S}^k= \{ i\in \mathcal S^k: E_i(\x^k) \geq \sigma M^k\}.
\end{equation}
 We denote by $\texttt{c}_{\mathcal S^k}$ the  cardinality of $\mathcal S^k$  normalized to the overall number of variables (in our experiments, all sets $\mathcal S^k$ have the same cardinality, i.e., $\texttt{c}_{\mathcal S^k}=\texttt{c}_{\mathcal S}$, for all $k$).  We considered the following options for  $\sigma$ and $\texttt{c}_{\mathcal S}$: i) $\texttt{c}_{\mathcal S}=0.01, 0.1, 0.2, 0.5, 0.8$; ii)  $\sigma=0$, which leads to a \emph{fully parallel} pure random scheme   wherein at each iteration \emph{all} variables in $\hat{\mathcal S}^k$ are updated; and iii) different positive values of $\sigma$ ranging from $0.01$ to $0.5$, which corresponds to updating in a greedy manner only a subset of  the variables in $\hat{\mathcal S}^k$ (the smaller the $\sigma$ the larger the number of potential variables to be updated at each iteration). 
  We termed     Algorithm 1 with $\sigma=0$  ``Random FLEXible parallel Algorithm''  (RFLEXA), whereas the other instances with  $\sigma>0$  as ``Hybrid FLEXA'' (HyFLEXA). 

\noindent \textbf{Algorithms in the literature}: We compared our versions of (Hy)FLEXA
with the most representative  \emph{parallel}   random and deterministic  algorithms proposed in the literature to solve the \emph{convex} instance of Problem (1) (and thus also LASSO). More specifically, we consider the following schemes.

\noindent $\bullet$ {\bf  PCDM $\&$ PCDM2}: These are  (proximal) gradient-like parallel randomized BCD methods proposed in \cite{richtarik2012parallel} for convex optimization problems. 
Since the authors recommend to use PCDM instead of PCDM2 for LASSO problems, we do so (indeed, our experiments show that PCDM outperforms PCDM2). We simulated PCDM under different sampling rules and   we set the  parameters $\beta$ and $\omega$ as in \cite[Table 4]{richtarik2012parallel}, which guarantees convergence of the algorithm in \emph{expected value}.

\noindent $\bullet$ {\bf Hydra $\&$ Hydra$^2$}: Hydra is a parallel and distributed random gradient-like CDM,  proposed in \cite{takavc2013Hydra}, wherein different cores in parallel update a randomly chosen subset of variables from those they own; a closed form solution of the scalar updates is available.    Hydra$^2$   \cite{fercoq2014fast} is the accelerated version of Hydra; indeed,  in all our experiments, it outperformed Hydra; therefore, we will report the results only for  Hydra$^2$. The free parameter $\beta$ is set to   $\beta=2\beta_1^\ast$ (cf. Eq. (15) in \cite{takavc2013Hydra}), with $\sigma$ given by   Eq. (12) in \cite{takavc2013Hydra} (according to the authors, this seems one of the best choices for $\beta$).

\noindent $\bullet$ {\bf FLEXA}: This is the parallel deterministic scheme we proposed in \cite{FacSagScuTSP14,FaccSagScu_ICASSP14}. We use   FLEXA as a benchmark of deterministic algorithms, since it  has been shown in \cite{FacSagScuTSP14,FaccSagScu_ICASSP14} that it    outperforms current (parallel) first-order (accelerated) gradient-like schemes, including FISTA \cite{beck_teboulle_jis2009}, SparRSA \cite{wright2009sparse}, GRock \cite{peng2013parallel}, parallel BCD \cite{tseng2009coordinate}, and parallel ADMM.  The free parameters of FLEXA,  $\tau_i$ and   $\gamma$, are tuned as in \cite[Sec. VI.A]{FacSagScuTSP14}, whereas the set $\mathcal{S}^k$ is chosen as in \eqref{S^k_simulations}.

\noindent $\bullet$ {\bf Other algorithms}: We tested also other random algorithms, including  \emph{sequential} random BCD-like methods  and Shotgun \cite{bradley2011parallel}. However,  since they were not competitive, to not overcrowd the figures,    we do not report results for  these algorithms.

In all the experiments,  the data matrix $\A = [ \A_1 \,\cdots\,\A_P]$ of the LASSO problem is stored in a column-block manner, uniformly across the $P$ parallel processes.
Thus the computation of each product $\A\x$ (required to evaluate $\nabla F$) and the norm $\|\x\|_1$ (that is $G$) is divided into
the parallel jobs of computing $\A_i \x_i$ and $\|\x_i\|_1$, followed by a reduce operation.
Also, for all the algorithms, the initial point  was set to the zero vector.

\noindent \textbf{Numerical Tests}: We generated synthetic  LASSO problems  using the random generation technique proposed by Nesterov  \cite{nesterov2012gradient}, which we properly modified    following \cite{richtarik2012parallel}  to generate instances of the problem with  different levels of sparsity of the solution as well as density of the  data matrix $\A\in \mathbb{R}^{m\times n}$; we introduce  the following two control parameters: $\texttt{s}_{\A}=$ average $\%$ of nonzeros in each column of $\A$ (out of $m$); and   $\texttt{s}_{\text{sol}}=$ $\%$ of nonzeros in the solution (out of $n$). We tested the algorithms on two groups of LASSO problems,   $\A\in \mathbb{R}^{10^4\times 10^5}$  and  $\A\in \mathbb{R}^{10^5\times 10^6}$, and several degrees of density of $\A$ and sparsity of the solution, namely $\texttt{s}_{\text{sol}}=0.1\%, 1\%, 5\%, 15\%, 30\%$, and $\texttt{s}_{\A}=10\%, 30\%, 50\%, 70\%, 90\%$. Because of the space limitation, we report  next only the most representative results; we refer to \cite{Amir_Report14} for more details and experiments.  Results  for the LASSO instance with  100,000 variables are reported in Fig. \ref{fig1} and \ref{fig2}. Fig. \ref{fig1}  shows the behavior of HyFLEXA as a function of the design parameters $\sigma$ and $\texttt{c}_{\mathcal S}$, for different values of the solution sparsity ($\texttt{s}_{\text{sol}})$, whereas in Fig. \ref{fig2} we compare the proposed RFLEXA and HyFLEXA with FLEXA, PCDM, and  Hydra$^2$, for different values of $\texttt{s}_{\text{sol}}$ and $\texttt{s}_{\A}$ (ranging from ``low'' dense  matrices  and ``high'' sparse solutions to ``high'' dense  matrices  and ``low'' sparse solutions).  Finally, in Fig. \ref{fig3} we consider   larger problems with $1$M   variables.  In all the figures,  we plot the relative error $\texttt{re}(\x) \triangleq  (V(\x)-V^*)/V^*$ versus the CPU time, where
 $V^*$ is the optimal value of the objective function $V$ (in our experiments $V^*$ is known). All  the curves are  averaged over ten independent random realizations. Note that  the CPU time includes communication times  and the initial time needed by the methods to perform all pre-iterations computations
(this explains why the  curves associated with   Hydra$^2$ start  after the others; in fact  Hydra$^2$ requires some  nontrivial computations to estimates $\beta$). Given Fig. \ref{fig1}-\ref{fig3}, the following comments are in order. \vspace{-0.2cm}
 
\begin{figure}[h]
\centering
        \begin{subfigure}[ ]{0.28\textwidth}
\hspace*{-0.6cm}
                \includegraphics[width=\textwidth]{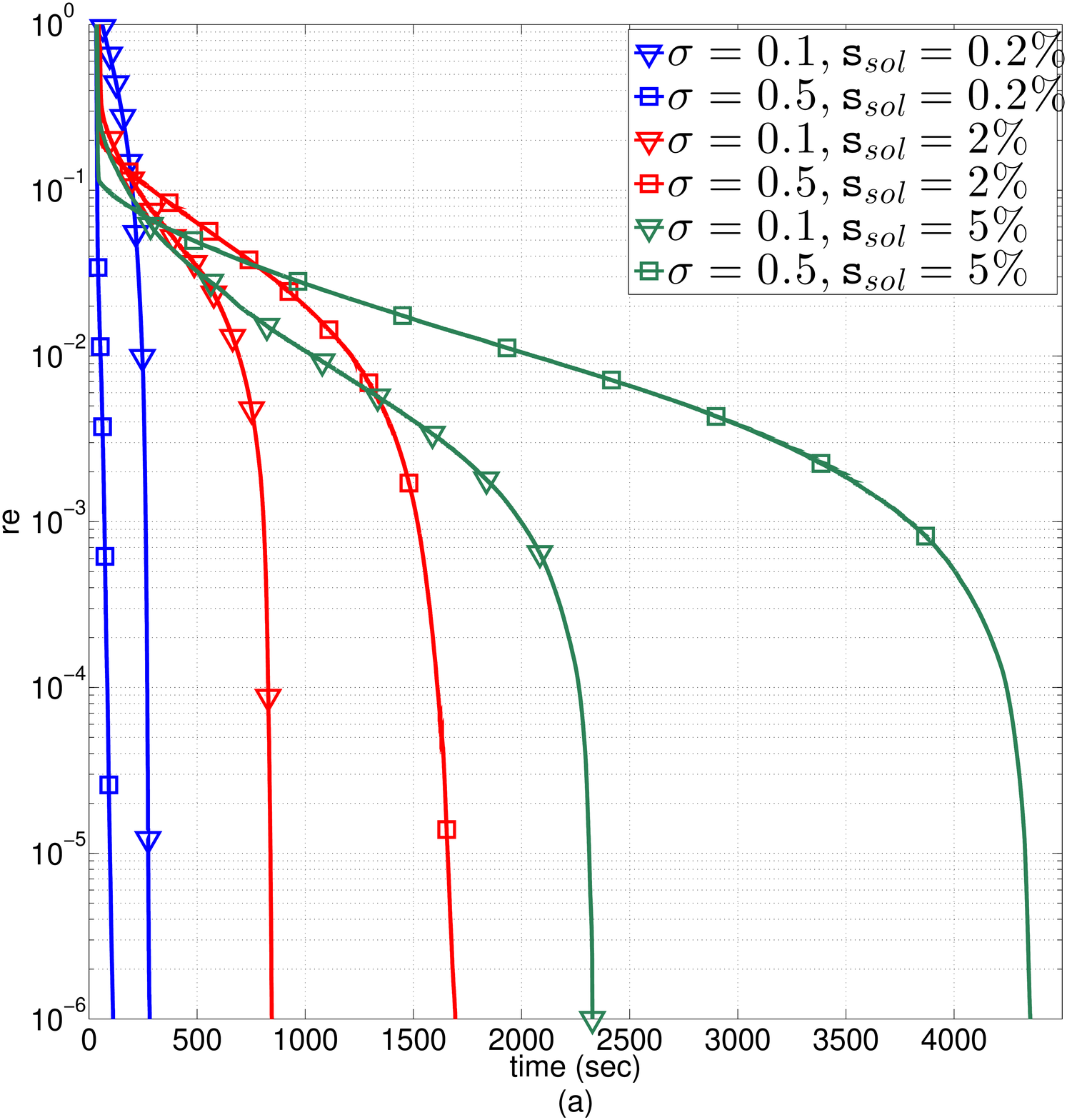}
        \end{subfigure}
        \begin{subfigure}[ ]{0.28\textwidth}
\hspace{-0.8cm}
               \hspace{-0.3cm}  \includegraphics[width=\textwidth]{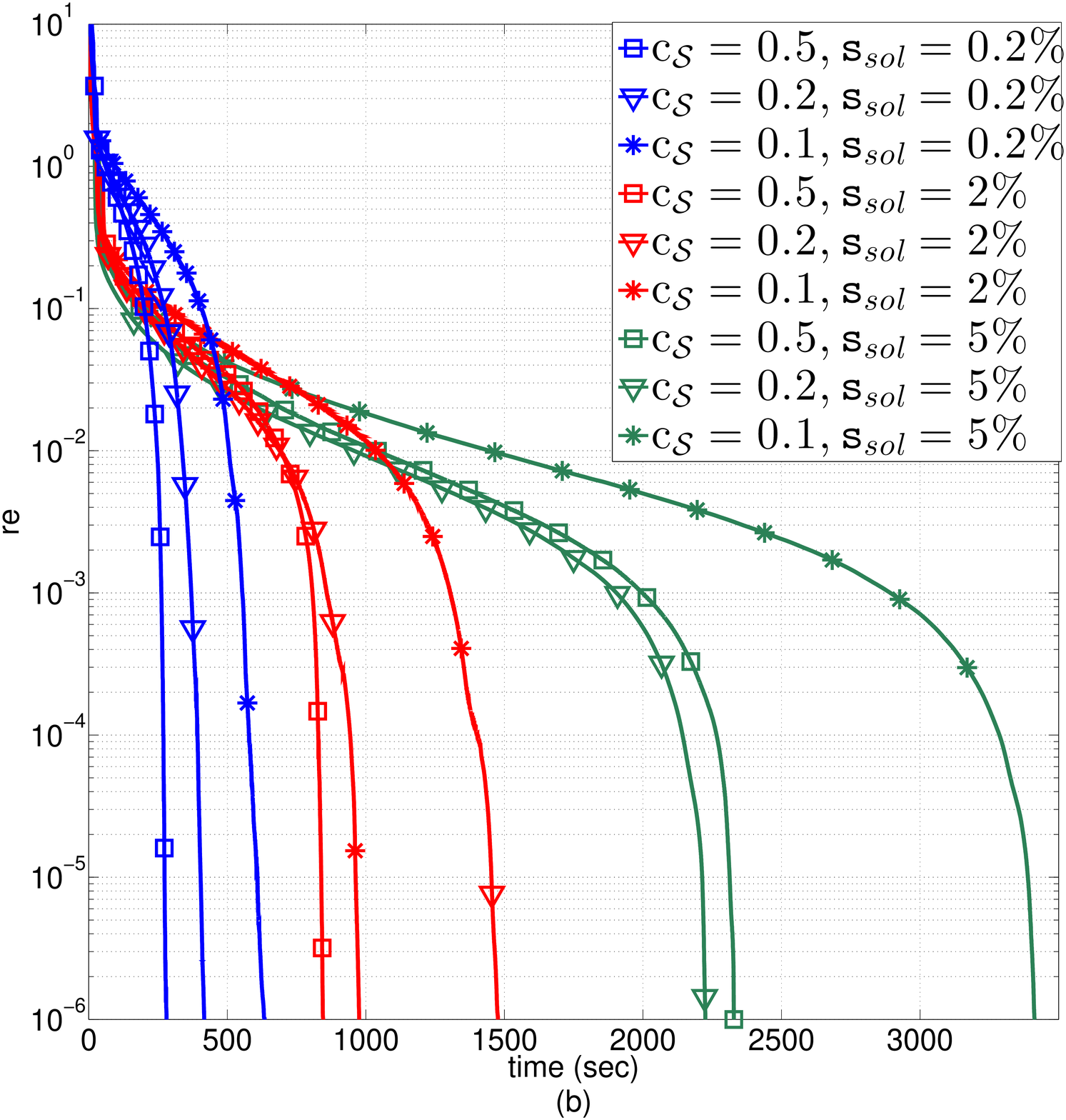}
        \end{subfigure}\vspace{-0.3cm}
        \caption{HyFLEXA for different values of $\texttt{c}_{\mathcal S}$ and $\sigma$:  Relative error vs. time; $\texttt{s}_{\text{sol}}=0.2\%, 2\%, 5\%$, $\texttt{s}_{\A}=70\%$, 100.000 variables, NU sampling, 8 cores; (a) $\texttt{c}_{\mathcal S}=0.5$, and $\sigma =0.1, 0.5$ - (b) $\sigma =0.5$, and $\texttt{c}_{\mathcal S}=0.1, 0.2, 0.5$.\label{fig1}}\vspace{-0.2cm}
        \end{figure}
        
\begin{figure}[]
\begin{subfigure}[]{\textwidth}
         { \vspace{-1.5cm} \hspace{0.8cm}   \includegraphics[width=0.35\textwidth]{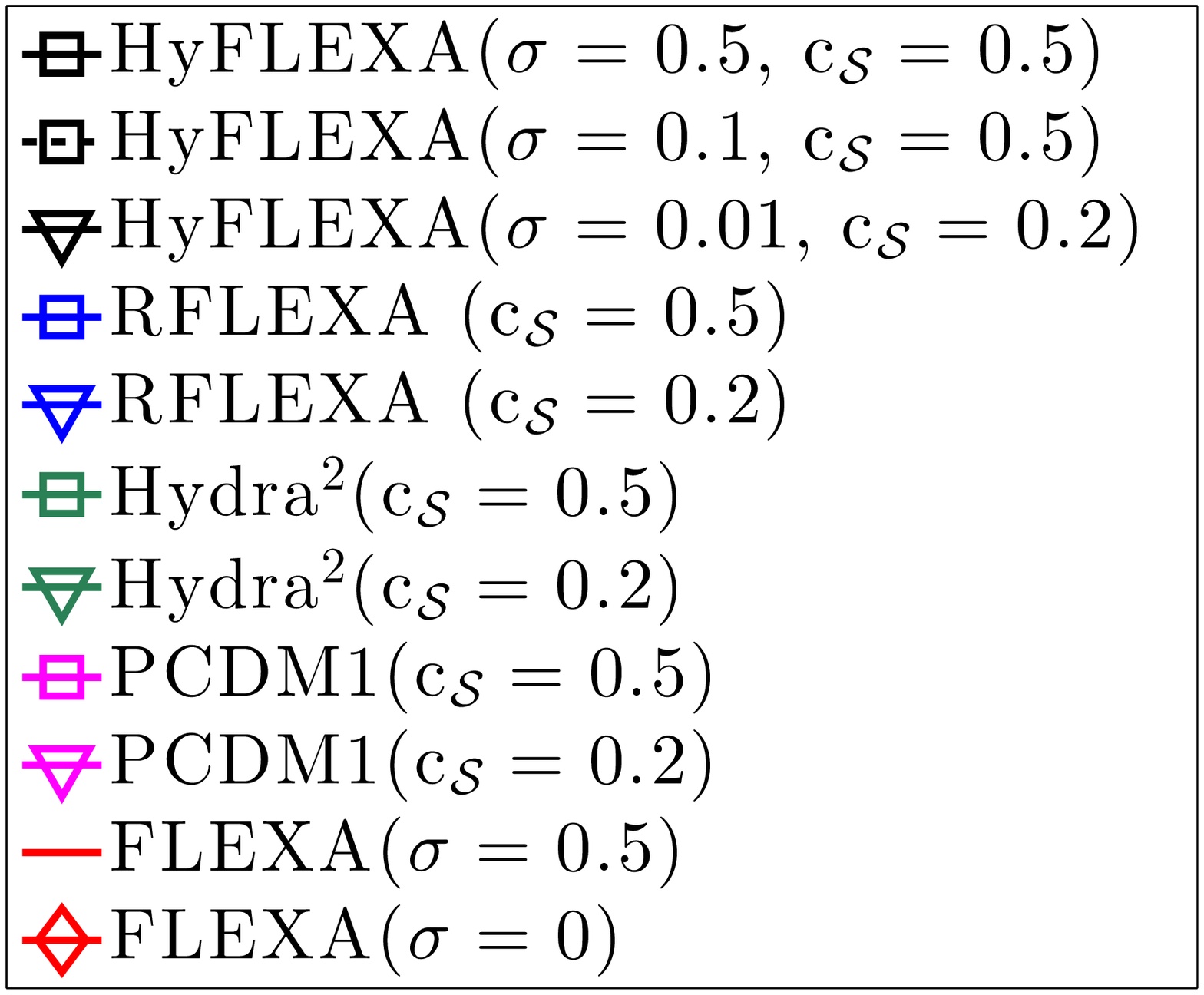}\vspace{-0.8cm}}\end{subfigure}  
      \begin{subfigure}[]{0.27\textwidth}
            \hspace{-0.4cm}   \includegraphics[width=\textwidth]{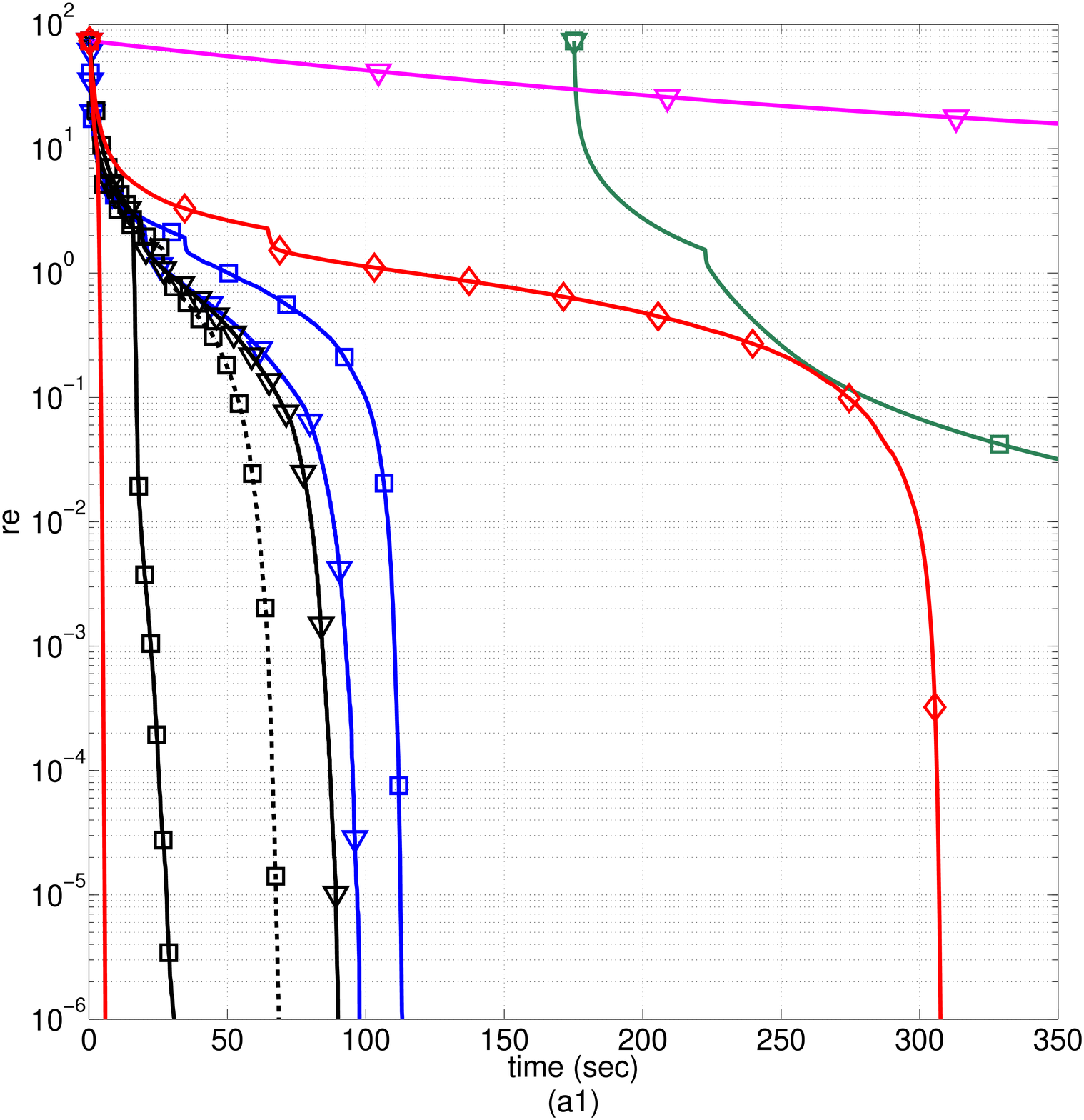}
      \end{subfigure}
      \begin{subfigure}[]{0.27\textwidth}
               \hspace{-0.8cm} \includegraphics[width=\textwidth]{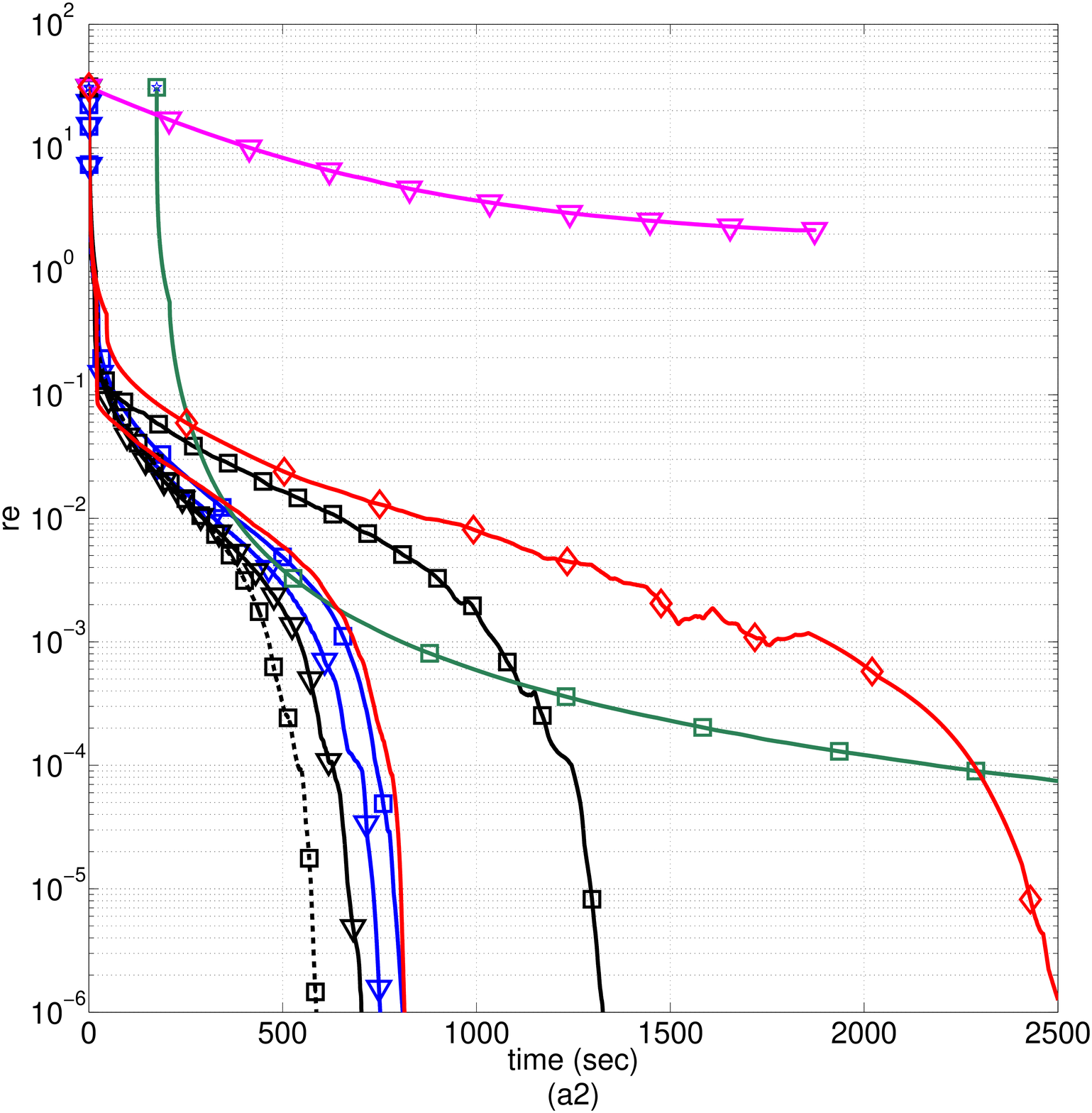}
      \end{subfigure}  
      \begin{subfigure}[]{0.27\textwidth}
            \hspace{-0.4cm}   \includegraphics[width=\textwidth]{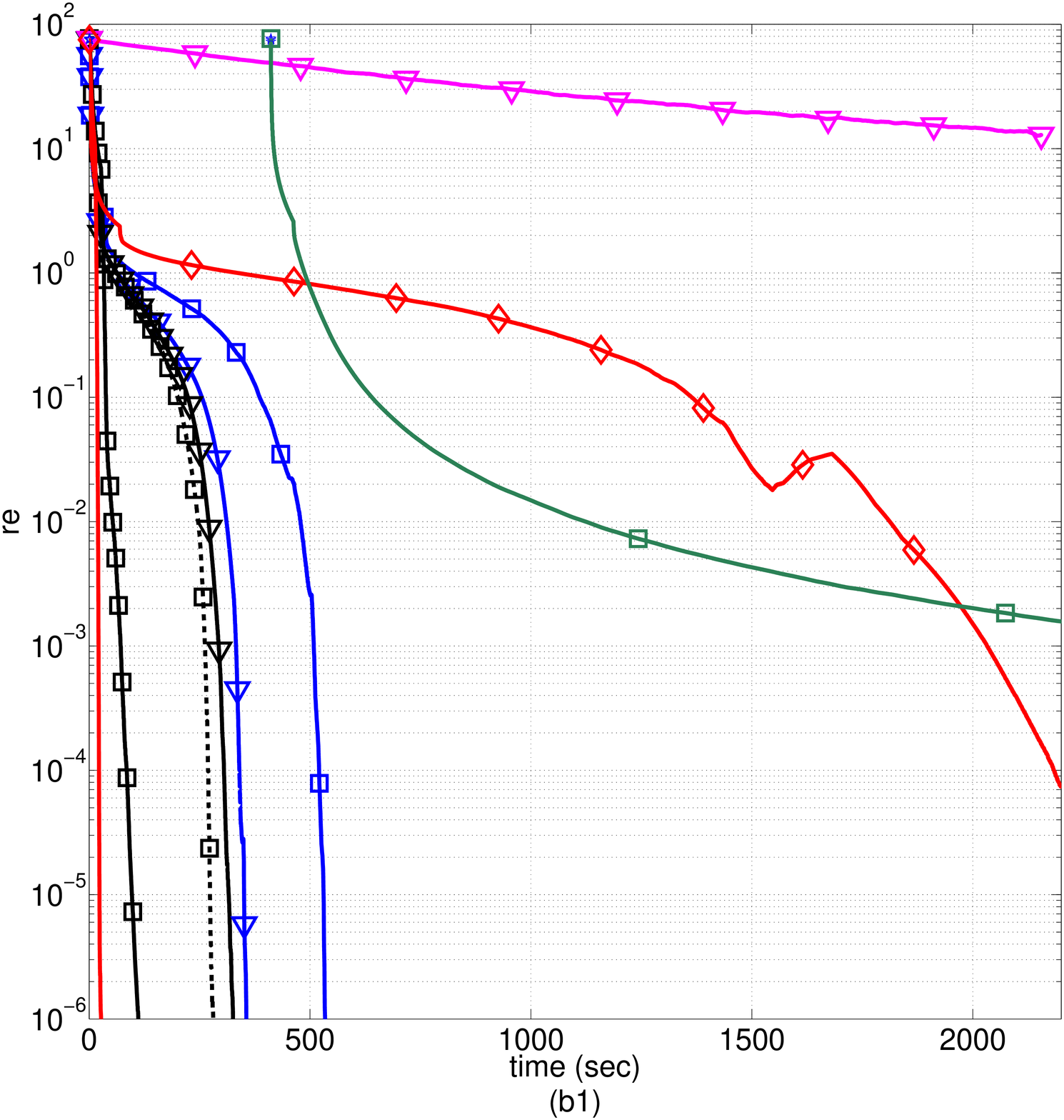}
      \end{subfigure}
      \begin{subfigure}[]{0.27\textwidth}
               \hspace{-0.8cm} \includegraphics[width=\textwidth]{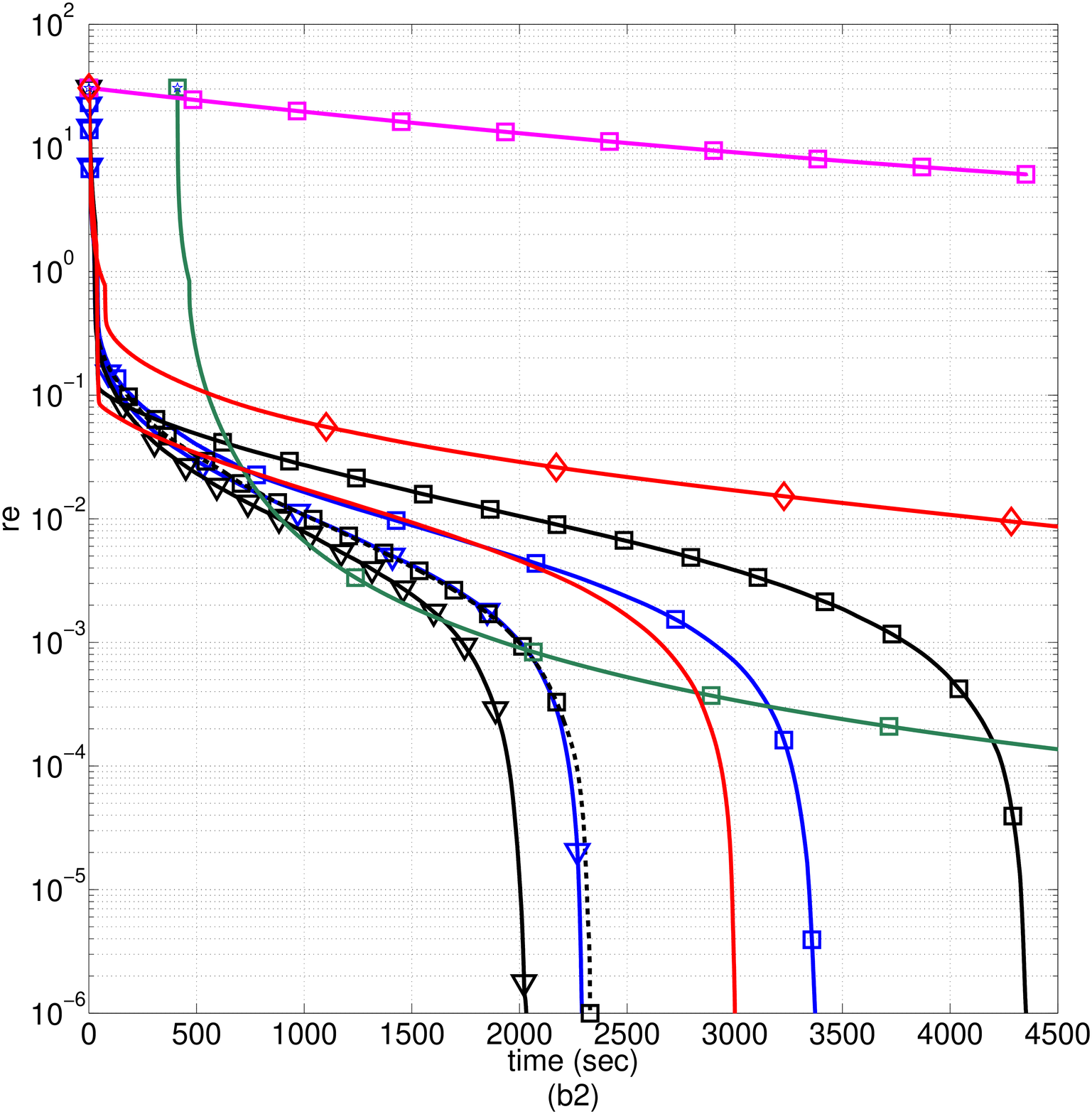}
      \end{subfigure}  
      \begin{subfigure}[]{0.27\textwidth}
            \hspace{-0.4cm}    \includegraphics[width=\textwidth]{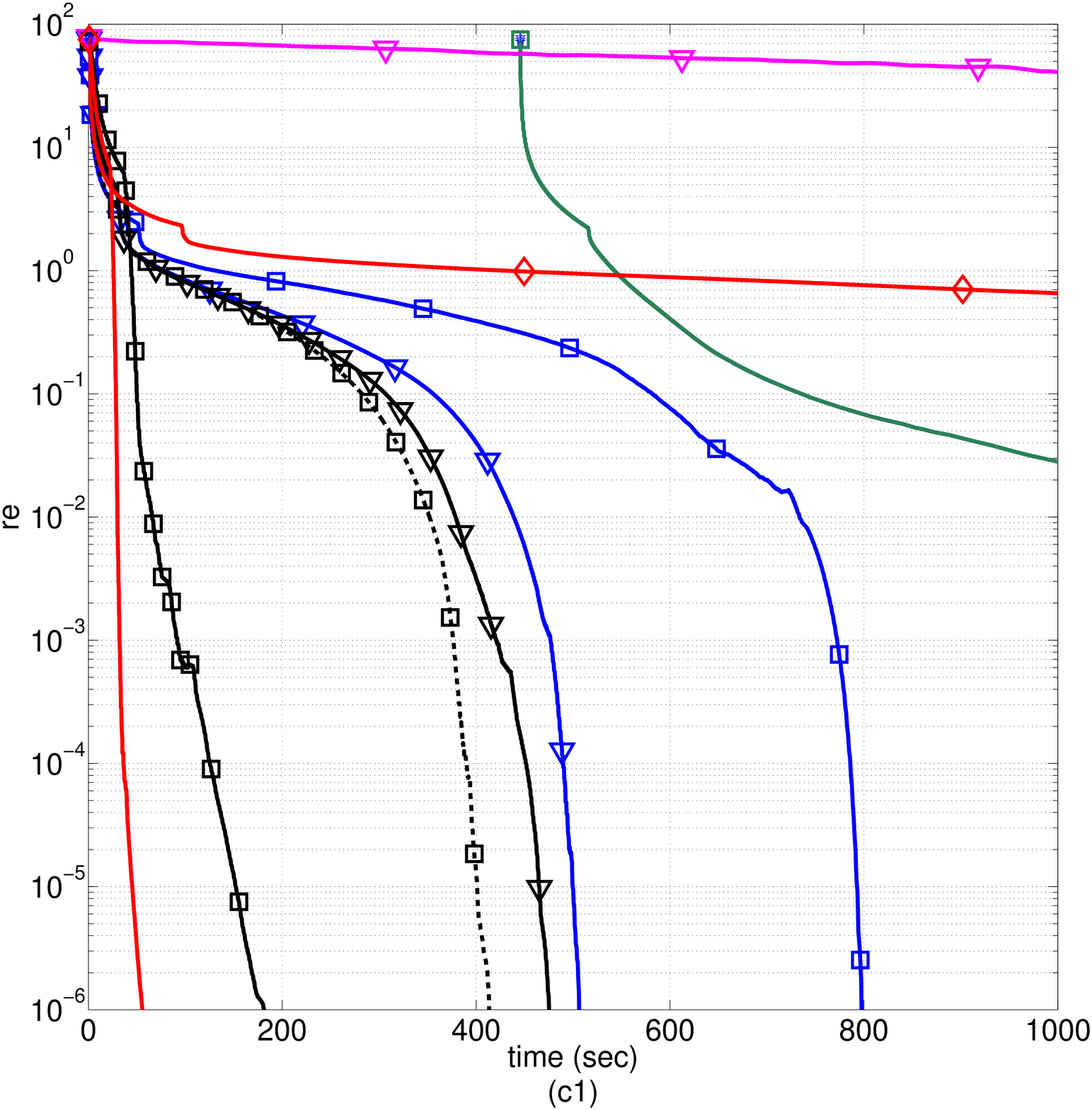}
      \end{subfigure}
      \begin{subfigure}[]{0.27\textwidth}
               \hspace{-0.8cm} \includegraphics[width=\textwidth]{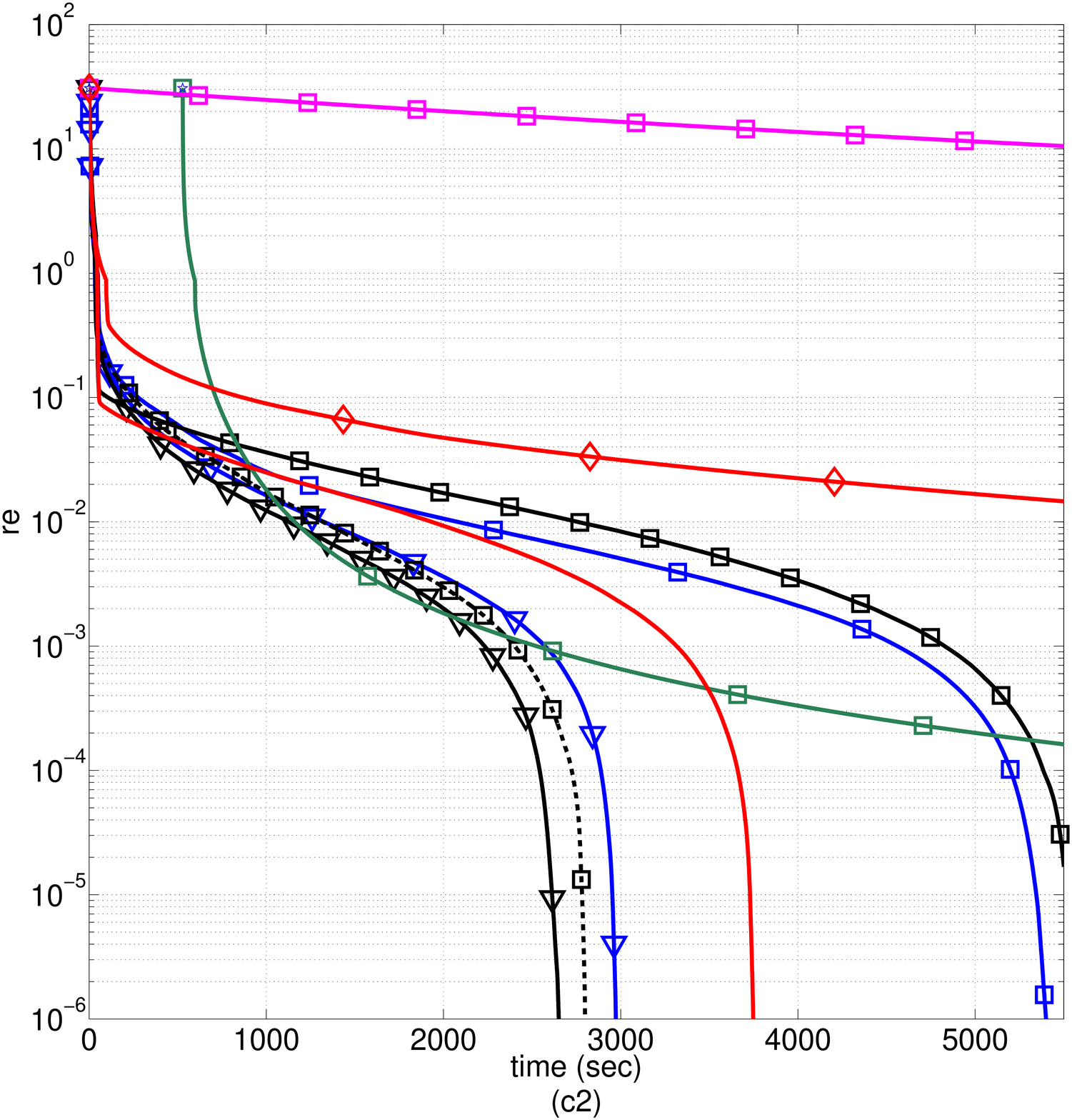}
      \end{subfigure}\vspace{-0.3cm}
\caption{LASSO with 100.000 variables, 8 cores; Relative error vs. time for:
(a1) $\texttt{s}_{\A}=30\%$ and $\texttt{s}_{\text{sol}}=0.2\%$ - (a2)  $\texttt{s}_{\A}=30\%$ and $\texttt{s}_{\text{sol}}=5\%$ - (b1)  $\texttt{s}_{\A}=70\%$ and $\texttt{s}_{\text{sol}}=0.2\%$ - (b2)  $\texttt{s}_{\A}=70\%$ and $\texttt{s}_{\text{sol}}=5\%$ - (c1)  $\texttt{s}_{\A}=90\%$ and $\texttt{s}_{\text{sol}}=0.2\%$ - (c2)  $\texttt{s}_{\A}=90\%$ and $\texttt{s}_{\text{sol}}=5\%$. \label{fig2}}\vspace{-0.5cm}
\end{figure}

\begin{figure}[t]
 \vspace{-0.3cm}
\centering
        \begin{subfigure}[ ]{0.28\textwidth}
\hspace*{-0.6cm}
                \includegraphics[width=\textwidth]{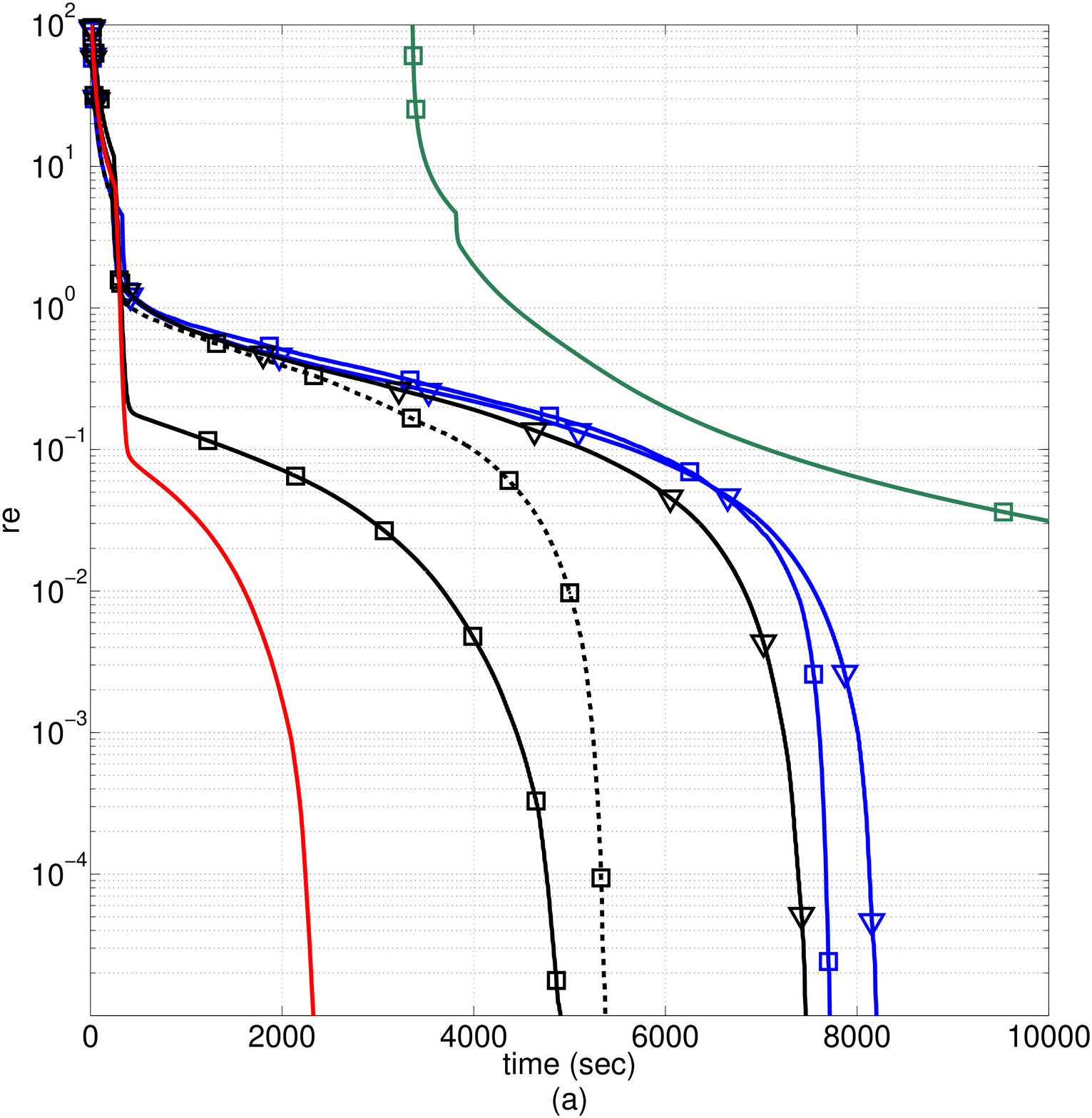}
        \end{subfigure}
        \begin{subfigure}[ ]{0.28\textwidth}
\hspace{-0.8cm}
               \hspace{-0.3cm}  \includegraphics[width=\textwidth]{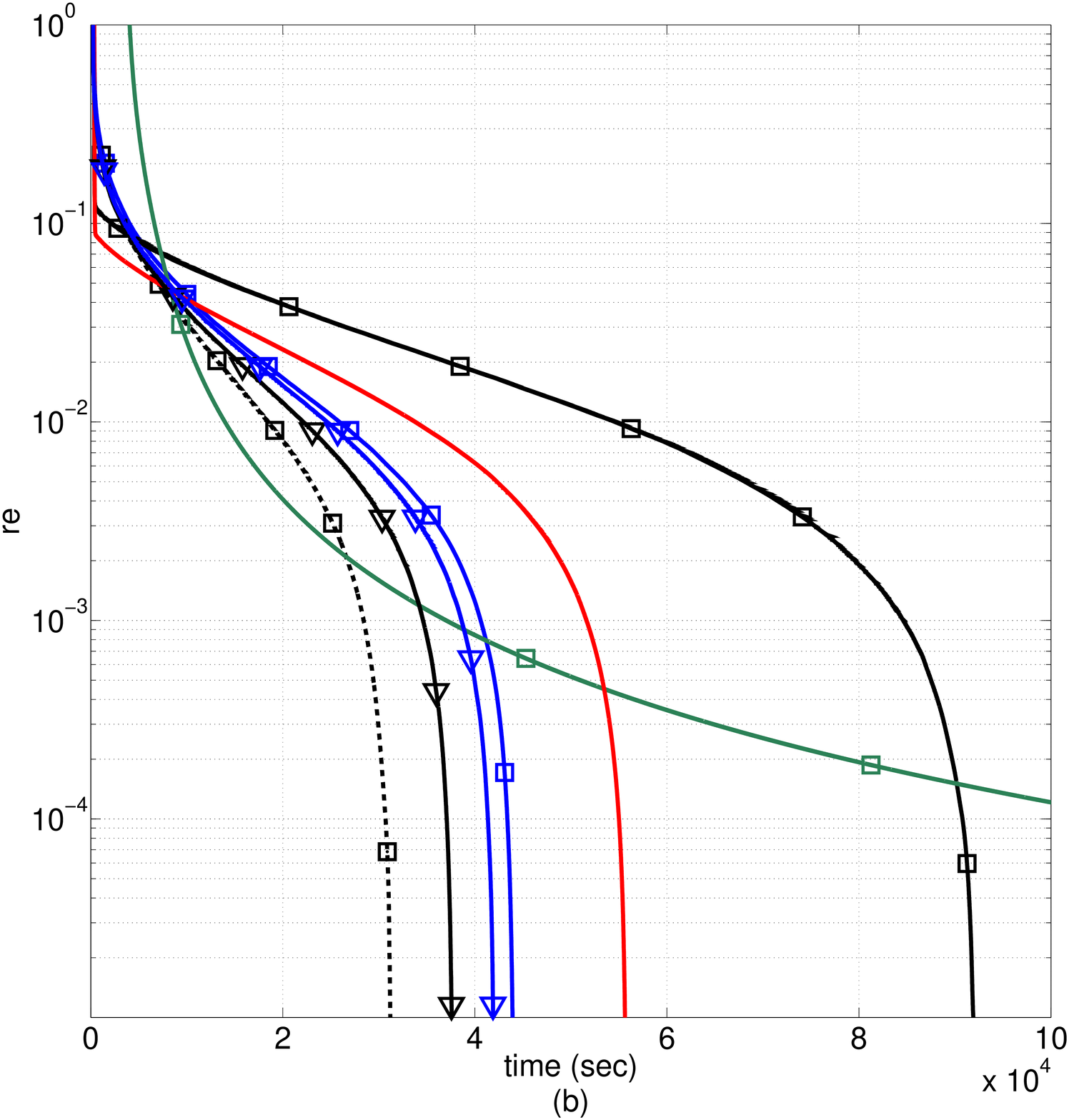}
        \end{subfigure}\vspace{-0.2cm}
        \caption{LASSO with  1M variables, $\texttt{s}_{\A}=10\%$, 16 cores; Relative error vs. time for: (a)$\texttt{s}_{\text{sol}}=1\%$ - (b)  $\texttt{s}_{\text{sol}}=5\%$. The legend is as in Fig. 2.
        \label{fig3}}\vspace{-0.5cm}
        \end{figure}

\noindent \emph{HyFLEXA: On the choice of $(\texttt{c}_{\cal S},\sigma)$, and the sampling strategy.} All the  experiments (including those that we cannot report here because of lack of space) show the following trend in the behavior of HyFLEXA as a function of  $(\texttt{c}_{\cal S}, \sigma)$. For ``low'' density problems (``low'' $\texttt{s}_{\text{sol}}$ and   $\texttt{s}_{\A}$), ``large'' pairs  $(\texttt{c}_{\cal S},\sigma)$ are  preferable, which   corresponds to updating at each iteration    only \emph{some}  variables by performing a  (heavy) greedy search over  a \emph{sizable} amount of variables. This is in agreement with \cite{FacSagScuTSP14} (cf. Remark 5): by the greedy selection, Algorithm 1 is able to identify those variables that will be zero at the a solution; therefore updating only variables that we have ``strong'' reason to believe will not be zero at a solution is  a better strategy than updating them all, especially if the solutions are very sparse.  Note that this behavior can be obtained    using  either   ``large''   or   ``small'' $(\texttt{c}_{\cal S},\sigma)$. However, in the case of ``low'' dense problems,  the former strategy outperforms the latter. We observed that this is mainly due to the fact that when $\texttt{s}_{\A}$ is ``small'',  estimating  $\hat{x}_i$ (computing the products $\A^T\A$) is computationally affordable, and thus performing a  greedy search over more variables 
enhances the practical  convergence. 
 When the sparsity of the solution decreases and/or the density of $\A$ increases (``large'' $\texttt{s}_{\A}$ and/or $\texttt{s}_{\text{sol}}$),  one can see from the figures that 
  ``smaller'' values of $(\texttt{c}_{\cal S}, \sigma)$ are more effective than  larger ones, which corresponds to  using  a ``less  aggressive'' greedy selection  while searching over a smaller pool of variables. In fact, when $\A$ is  dense, computing all $\hat{x}_i$ might be prohibitive and thus nullify the potential benefits of a greedy procedure. 
 For instance, it follows from Fig. 1-3 that, as the density of the  solution ($\texttt{s}_{\text{sol}}$) increases   the preferable  choice for $(\texttt{c}_{\cal S}, \sigma)$ progressively moves from $(0.5, 0.5)$ to $(0.2, 0.01)$, with both $\texttt{c}_{\cal S}$ and  $\sigma$ decreasing. Interesting, a tuning that works quite well in practice for all the classes of problems we simulated (different densities of $\A$, solution sparsity, number of cores, etc.)  is $(\texttt{c}_{\cal S}, \sigma)=(0.5, 0.1)$, which seems to strike a good balance between not updating variables that are probably zero at the optimum and nevertheless update a sizable amount of variables when needed in order to enhance convergence.. 

As a final remark, we report that, according to our experiments, the most effective sampling rule    among  U, DU, NU, and NS is  the NU (which is actually the one the figures refers to); NS becomes competitive only when the solutions are very sparse, see \cite{Amir_Report14} for a detailed comparison of the different rules.

\noindent \emph{Comparison of the algorithms}. 
For low dense matrices $\A$  and very sparse solutions, FLEXA $\sigma=0.5$ is faster than   its random counterparts (RFLEXA and HyFLEXA) as well as its fully parallel version,  FLEXA $\sigma=0$ [see Fig 2 a1), b1) c1) and Fig. 3a)]. Nevertheless, HyFLEXA [with $(\texttt{c}_{\cal S}, \sigma)=(0.5,0.5)]$ remains close. As already pointed out, this is mainly due to   the fact that in these scenarios i) estimating \emph{all} $\hat{x}_i$ is computationally cheap (and thus performing a greedy selection over a sizable set of variable is beneficial, see Fig. 1); and ii) updating only some variables at each iteration is more effective than updating all (FLEXA $\sigma=0.5$  outperforms FLEXA $\sigma=0$).  However, as the density of $\A$ and/or the size of the problem increase, computing all the products $[\A^T\A]_{ii}$ (required to estimate $\hat{x}_i$) becomes too costly; this is when  a random selection of the variables becomes beneficial: indeed, RFLEXA and HyFLEXA consistently outperform FLEXA [see Fig 2 a2), b2) c2) and Fig. 3b)]. Among the random algorithms, Hydra$^2$ is capable to approach relatively fast low accuracy, especially when the solution is not too sparse, but has difficulties in reaching high accuracy. RFLEXA and HyFLEXA are always much  faster than  current state-of-the-art schemes (PCDM and  Hydra$^2$), especially if high accuracy of the solutions is required. Between RFLEXA and HyFLEXA (with the same $\texttt{c}_{\cal S}$), the latter  consistently outperforms the former (about up to five time faster),  with a gap that is more significant when  solutions are sparse. This provides a solid evidence of the effectiveness of the proposed hybrid random/greedy selection method. 

In conclusion, our experiments  indicate  that the proposed framework   leads to very efficient and practical solution  methods for large  and very large-scale (LASSO)    
 problems, with the flexibility to adapt to many different problem characteristics.\vspace{-0.2cm}

\section{Conclusions}\label{sec:conclusions} 
We  proposed a highly parallelizable  hybrid random/deterministic decomposition algorithm for  the minimization of the sum of a possibly noncovex  differentiable function $F$ and a possibily nonsmooth nonseparable  convex function $G$. 
The proposed  framework is the first scheme enjoying all the following features: i) it allows for pure greedy, pure random, or  mixed random/greedy updates of the variables, all converging under the same unified set of convergence conditions; ii) it can tackle via parallel updates also nonseparable convex functions $G$; iii) it can deal with nonconvex nonseparable $F$; iv) it is parallel; v) it can incorporate both first-order or higher-order information; and vi) it can use inexact solutions.  Our preliminary experiments on LASSO problems  showed the superiority of the proposed scheme with respect to state-of-the-art random and deterministic algorithms. 
Experiments on more varied classes of problems   are the subject of our current research.\vspace{-0.2cm}

\section{ Acknowledgments}\vspace{-0.1cm}
The authors are   very  grateful to Prof. Peter Richt\`{a}rik  for his invaluable comments; we also thank  Dr. Martin Tak\'{a}\v{c} and  Prof. Peter Richt\`{a}rik   for providing the C++ code of PCDM and Hydra$^2$ (that we modified in order to use the  MPI library).

The work of Daneshmand and Scutari was supported by the USA NSF Grants CMS 1218717 and  CAREER Award No. 1254739.
The work of Facchinei was supported by the MIUR project PLATINO (Grant Agreement n. PON01\_01007).
The work of Kungurtsev was supported by the European Social Fund
under the Grant  CZ.1.07/2.3.00/30.0034.\vspace{-0.1cm}

\appendix\vspace{-0.1cm}
\section{Appendix: Proof of Theorem \ref{Theorem_convergence_inexact_Jacobi} and 3}
We first introduce some preliminary results instrumental to prove both Theorem \ref{Theorem_convergence_inexact_Jacobi} and Theorem 3.  Given $\hat{\mathcal S}^k\subseteq \mathcal N$ and $\x\triangleq (\x_i)_{i\in \mathcal N}$,  for   notational simplicity, we will  denote by $(\mathbf x)_{\hat{\mathcal S}^k}^k$ (or interchangeably $\mathbf x_{\hat{\mathcal{S}}^k}^k$) the vector whose component  $i$ is equal to $\x_i$ if $i\in \hat{\mathcal{S}}^k$, and zero otherwise. With a slight abuse of notation we will also use $(\x_i,\y_{-i})$ to denote the ordered tuple $(\y_1,\ldots, \y_{i-1},\x_i,\y_{i+1},\ldots, \y_N)$; similarly $(\x_i,\x_j,\y_{-(i,j)})$, with $i<j$ stands for $(\y_1,\ldots, \y_{i-1},\x_i,\y_{i+1},\ldots, \y_{j-1},\x_j,\y_{j+1},\ldots, \y_N)$.\vspace{-0.1cm}

\subsection{On the random sampling and its properties}\label{App:random_properties}
We introduce some properties associated with the random sampling rules $\boldsymbol{\mathcal{S}}^k$ satisfying assumption A6. A key role in our proofs is played by the following  random set: let  $\{\x^k\}$ be the sequence generated by Algorithm 1, and\vspace{-0.1cm}
\begin{equation}\label{i_max_k}
i^k_{\text{mx}}  =\underset{i\in\{1,...,N\}}{\mbox{argmax}\,} ||\widehat{\x}_i(\x^k)-\x^k_i||,
\end{equation}
define the set   $\mathcal{K}_{\text{mx}}$ as\vspace{-0.1cm}
\begin{equation}\label{K_max}
\mathcal{K}_\text{mx}\triangleq \left\{k\in \mathbb{N}_+\,:\, i^k_\text{mx}\in \boldsymbol{\mathcal{S}}^k\right\}.
\end{equation}

The key properties of this set are summarized in the following two lemmata.

\begin{lemma}[Infinite cardinality] Given the set  $\mathcal{K}_{\text{mx}}$ as in \eqref{K_max}, it holds that\vspace{-0.2cm} \[ \mathbb{P} \left( \left |\mathcal{K}_{\text{\emph{mx}}}\right|=\infty  \right)=1, \] where $\left|\mathcal{K}_{\text{\emph{mx}}}\right|$ denotes the cardinality of $\mathcal{K}_\text{\emph{mx}}$.\vspace{-0.1cm}
\end{lemma}
\begin{proof}
Suppose that the statement of the lemma is not true.
Then, with positive probability, there must exist some $\bar{k}$ such that for $k\ge \bar{k}$,
$i^k_\text{mx}\notin \boldsymbol{\mathcal{S}}^k$.
 But we can write
 \begin{align*}
 &\mathbb{P}\left(\left\{ i^k_{\text{mx}}\notin \boldsymbol{\mathcal{S}}^k \right\}_{k\geq \bar{k}}\right) \\
 & \;  =
\underset{k\geq\bar{k}}{\Pi}
  \mathbb{P}\left(i^k_{\text{mx}}\notin \boldsymbol{\mathcal{S}}^k \left |\right.
 ( i^{\bar k}_{\text{mx}}\notin \boldsymbol{\mathcal{S}}^{\bar k}), \;  \dots \, , \;
 ( i^{k-1}_{\text{mx}}\notin \boldsymbol{\mathcal{S}}^{k-1}) \right) \\
 & \; \leq \lim_{k \to \infty}  (1-p)^{k-\bar k} =0.
 \end{align*}
 where the inequality follows by A6 and the independence of the events. But this obviously gives a contradiction and concludes the proof.
\end{proof}


\begin{lemma}\label{lemma:almost_surely_nonsum} Let   $\{\gamma^k\}$ be a sequence satisfying assumptions i)-iii) of Theorem \ref{Theorem_convergence_inexact_Jacobi}. Then it holds that \begin{equation}\label{almost_surely_nonsum}
\mathbb{P}\left( \sum_{k\in \mathcal{K}_\text{\emph{mx}}} \gamma^k < \infty \right)=0.
\end{equation}
\end{lemma}
\begin{proof}

It holds that,\vspace{-0.2cm}
\begin{eqnarray}\label{union_of_Pn}
\mathbb{P}\left(\sum_{k\in \Km} \gamma^k < \infty\right)&\le &
\mathbb{P}\left(\bigcup_{n\in \mathbb{N}}\sum_{k\in \Km} \gamma^k < n\right)\medskip\nonumber\\ &\le&
\sum_{n\in\mathbb{N}} \mathbb{P}\left(\sum_{k\in\Km} \gamma^k < n\right).\nonumber
\end{eqnarray}
To prove the lemma, it is then sufficient to show that  $\mathbb{P}\left(\sum_{k\in\Km} \gamma^k < n\right)=0$, as proved next.

Define $\hat{K}_i$, with $i\in \mathbb{N}_+$,  as the smallest index $\hat{K}_i$ such that \vspace{-0.1cm} \begin{equation}\label{K_i_hat}\sum_{j=0}^{\hat{K}_i} \gamma^j \ge i\cdot n.\vspace{-0.2cm}
\end{equation}   Note that  since   $\sum_{k=0}^\infty \gamma^k=+\infty$,
$\hat{K}_i$  is
well-defined for all $i$ and $\lim_{i\to \infty} \hat K_i = +\infty$.
 For any $n\in\mathbb{N}$, it holds:
 \begin{equation} \label{prob_gamma_set1}
\begin{aligned}
 &\mathbb{P}\left(\sum_{k\in \Km} \gamma^k < n\right)
 = \mathbb{P}\left(\bigcap_{m\in\mathbb{N}} \left( \sum_{k\in \Km}^m \gamma^k < n\right)\right)\qquad{}\medskip\\
 &  = \lim_{m\to \infty} \mathbb{P}\left(\sum_{k\in \Km}^m \gamma^k < n\right)
 {=} \lim_{i\to \infty} \mathbb{P}\left(\sum_{k\in \Km}^{\hat{K}_i} \gamma^k < n\right)\medskip\\
&= \lim_{i\to \infty}   \left[ \mathbb{P}\left(\sum_{k\in \Km}^{\hat{K}_i} \gamma^k < n, \,\,\,
|\Km\cap [0,\hat{K}_i]| < \frac{\hat{K}_i}{\sqrt{i}} \right)\right.\smallskip\\
&\qquad+ \left. \mathbb{P}\left(\sum_{k\in \Km}^{\hat{K}_i} \gamma^k < n, \,\,\,
|\Km\cap [0,\hat{K}_i]| \geq \frac{\hat{K}_i}{\sqrt{i}} \right)\right]\medskip\\
&\leq  \lim_{i\to \infty}  \left[ \underset{\mbox{term I}}{\underbrace{\mathbb{P}\left(
|\Km\cap [0,\hat{K}_i]| < \frac{\hat{K}_i}{\sqrt{i}} \right)}}\right.\\
& \qquad+ \left.
\underset{\mbox{term II}}{\underbrace{\mathbb{P}\left(\sum_{k\in \Km}^{\hat{K}_i} \gamma^k < n, \,\,\,
|\Km\cap [0,\hat{K}_i]| \geq \frac{\hat{K}_i}{\sqrt{i}} \right)}}\right].
 \end{aligned}
 \end{equation}
Let us bound next ``term I'' and ``term II'' separately.  

 \noindent \emph{Term I}: We have\vspace{-0.2cm}
  \begin{equation} \label{termI}
\begin{aligned}
&\mathbb{P}\left(
|\Km\cap [0,\hat{K}_i]| < \frac{\hat{K}_i}{\sqrt{i}} \right) \overset{(a)}{=}   \mathbb{P}\left(\sum_{k=0}^{\hat{K}_i} \X_k < \frac{\hat{K}_i}{\sqrt{i}} \right) \medskip\\
& {\leq}\, \mathbb{P}\left(\left|\sum_{k=0}^{\hat{K}_i} \X_k - \sum_{k=0}^{\hat{K}_i}p_k\right|> \sum_{k=0}^{\hat{K}_i}p_k-\frac{\hat{K}_i}{\sqrt{i}} \right)\bigskip\\
&\overset{(b)}{\leq}\,  \left(\frac{\sqrt{\sum_{k=0}^{\hat{K}_i} p_k (1-p_k)}}{\sum_{k=0}^{\hat{K}_i}p_k -\frac{\hat{K}_i}{\sqrt{i}}}\right)^2\overset{(c)}{\leq}\,  \left(\frac{\sqrt{\hat{K}_i}}{\hat{K}_i\left(p-\frac{1}{\sqrt{i}}\right)} \right)^2\bigskip\\
&=\,    \left(\frac{1}{\sqrt{\hat{K}_i}\left(p-\frac{1}{\sqrt{i}}\right)} \right)^2 \underset{i\to\infty}{\longrightarrow}0
 \end{aligned} \vspace{-0.2cm}
 \end{equation} 
where:

\noindent (a): $\X_0,\ldots, \X_{\hat{K}_i}$ are independent Bernoulli random variables, with parameter  $p_{{k}}\triangleq \mathbb{P}(k\in \Km)$. Note that, due to A6, $p_k\geq p$, for all $k$;

\noindent (b): it follows from  Chebyshev's inequality;

\noindent (c): we used the bounds $\sum_{k=0}^{\hat{K}_i} p_k (1-p_k)\leq \hat{K}_i$ and $\sum_{k=0}^{\hat{K}_i}p_k\geq p \hat{K}_i$.\smallskip

\noindent \emph{Term II}: Let us rewrite term II as
  \begin{equation} \label{termII}
  \hspace{-0.2cm}
\begin{aligned}
&\mathbb{P}\left(\frac{\sum_{k\in \Km}^{\hat{K}_i} \gamma^k}{|\Km\cap [0,\hat{K}_i]|} < \frac{n}{|\Km\cap [0,\hat{K}_i]|}\right.\smallskip\\
&\qquad \left.\left|\,
|\Km\cap [0,\hat{K}_i]| \geq \frac{\hat{K}_i}{\sqrt{i}}\right. \right)\smallskip
\cdot\mathbb{P}\left(|\Km\cap [0,\hat{K}_i]|\geq \frac{\hat{K}_i}{\sqrt{i}}\right)\medskip\\
& \overset{(a)}{\leq} \mathbb{P}\left(\frac{\sum_{k\in \Km}^{\hat{K}_i} \gamma^k}{|\Km\cap [0,\hat{K}_i]|} < \frac{n\,\sqrt{i}}{\hat{K}_i}\right.\left.\left|\,
|\Km\cap [0,\hat{K}_i]| \geq \frac{\hat{K}_i}{\sqrt{i}}\right. \right)\smallskip\\
&\quad \cdot\mathbb{P}\left(|\Km\cap [0,\hat{K}_i]|\geq \frac{\hat{K}_i}{\sqrt{i}}\right)\medskip\\
& \overset{(b)}{\leq}\, \mathbb{P}\left(\frac{\sum_{k\in \Km}^{\hat{K}_i} \gamma^k}{|\Km\cap [0,\hat{K}_i]|} < \frac{\sum_{k=0}^{\hat{K}_i} \gamma^k}{ \hat{K}_i\sqrt{i}}\right)\medskip\\
& \overset{(c)}{\leq}\, \mathbb{P}\left(\frac{{\sum_{k=0}^{\hat{K}_i} \gamma^k \X_k}}{\hat{K}_i} < \frac{\sum_{k=0}^{\hat{K}_i} \gamma^k}{ \hat{K}_i }\,\frac{1}{\sqrt{i}}\right)\medskip\\
& \,\,{\leq}\,\,\, \mathbb{P}\left(\left|\frac{{\sum_{k=0}^{\hat{K}_i} \gamma^k \X_k}}{\hat{K}} - \frac{\sum_{k=0}^{\hat{K}_i} \gamma^k\,p_k}{ \hat{K}_i}\right| \right.\smallskip\\
&\qquad\qquad \left. > \frac{\sum_{k=0}^{\hat{K}_i} \gamma^k\,p_k}{ \hat{K}_i}-\frac{\sum_{k=0}^{\hat{K}_i} \gamma^k}{ \hat{K}_i}\,\frac{1}{\sqrt{i}}\right)\medskip\\
& \,\,{\leq}\,\,\, \mathbb{P}\left(\left|\frac{{\sum_{k=0}^{\hat{K}_i} \gamma^k \X_k}}{\hat{K}_i} - \frac{\sum_{k=0}^{\hat{K}_i} \gamma^k\,p_k}{ \hat{K}_i}\right| \right. \left. > \left(p-\frac{1}{\sqrt{i}}\right)\frac{\sum_{k=0}^{\hat{K}_i} \gamma^k}{ \hat{K}_i}\right)\medskip\\
&  \overset{(d)}{\leq}   \left(\frac{\sqrt{{\sum_{k=0}^{\hat{K}_i} (\gamma^k)^2 \,p\,(1-p)}}} { \left(p-\frac{1}{\sqrt{i}}\right){\sum_{k=0}^{\hat{K}_i} \gamma^k}}\right)^2\leq   \left(\frac{\sqrt{{\sum_{k=0}^{\hat{K}_i} \gamma^k}}} { \left(p-\frac{1}{\sqrt{i}}\right){\sum_{k=0}^{\hat{K}_i} \gamma^k}}\right)^2\medskip\\
& = \left(\frac{1} { \left(p-\frac{1}{\sqrt{i}}\right)\sqrt{\sum_{k=0}^{\hat{K}_i} \gamma^k}}\right)^2 \underset{i\to\infty}{\longrightarrow}0,
\end{aligned}
 \end{equation}
where:

\noindent (a): we used  $|\Km\cap [0,\hat{K}_i]| \geq \frac{\hat{K}_i}{\sqrt{i}}$, by the conditioning event;

\noindent (b): it follows from \eqref{K_i_hat}, and $\mathbb{P}(A \bigcap B)\leq \mathbb{P}(A)$;

\noindent (c): $\X_{0},\ldots, \X_{\hat{K}_i}$ are independent Bernoulli random variables, with parameter  $p_{{k}}$. The bound is due to  $|\Km\cap [0,\hat{K}_i]| \leq \hat{K}_i$;

\noindent (d): it follows from the Chebyshev's inequality.

The desired result \eqref{almost_surely_nonsum} follows readily combining \eqref{prob_gamma_set1}, \eqref{termI}, and \eqref{termII}.
\end{proof}

\subsection{On the best-response map $\widehat{\mathbf{x}}(\bullet)$ and its properties}\label{sec:BR_properties}
We introduce now some key properties of the mapping $\widehat{\mathbf{x}}(\bullet)$ defined in (\ref{eq:decoupled_problem_i}). We also derive some bounds involving $\widehat{\mathbf{x}}(\bullet)$ along with the sequence $\{\x^k\}$ generated by Algorithm 1.

\begin{lemma}[\cite{FacSagScuTSP14}]\label{Prop_x_y} Consider  Problem (\ref{eq:problem 1}) under A1-A5, and F1-F3. Suppose that $G(\mathbf{x})$ is separable, i.e., $G(\mathbf{x})=\sum_i G_i(\x_i)$, with each $G_i$ convex on $X_i$.
Then the mapping  $X\ni\mathbf{y}\mapsto\widehat{\mathbf{x}}(\mathbf{y})$
is Lipschitz continuous on \emph{$X$}, i.e., there
exists a positive constant $\hat{{L}}$ such that\emph{
\begin{equation}
\left\Vert \widehat{\mathbf{x}}(\mathbf{y})-\widehat{\mathbf{x}}(\mathbf{z})\right\Vert \leq\,\hat{{L}}\,\left\Vert \mathbf{y}-\mathbf{z}\right\Vert ,\quad\forall\mathbf{y},\mathbf{z}\in X.\label{eq:Lipt_x_map}
\end{equation}
}
\end{lemma} \smallskip

\begin{lemma}\label{lem:sampledesc}
Let $\{\x^k\}$ be the sequence generated by Algorithm 1. For every  $k\in \mathcal{K}_{\text{mx}}$ and $\hat{\mathcal{S}}^k$ generated as  in   step S.3 of Algorithm 1, the following holds: there exists a positive constant $c_1$ such that,
\begin{equation}\label{eq:samplediff}
||\hat{\x}_{\hat{\mathcal{S}}^k}(\x^k)-\x^k_{\hat{\mathcal{S}}^k}|| \ge c_1\,||\hat{\x}(\x^k)-\x^k||.
\end{equation}
\end{lemma}
\begin{proof}
The following chain of inequalities holds:  \[
\begin{aligned}
&\left(\max_{i\in \mathcal{N}} \bar{s}_{i}\right) \left\|\hat{\x}_{\hat{\mathcal{S}}^k}(\x^k)-\x^k_{\hat{\mathcal{S}}^k}\right\|  \overset{(a)}{\ge} \bar{s}_{i^k_\rho}\, \left\|\hat{\x}_{i^k_\rho}(\x^k)-\x^k_{{i^k_\rho}}\right\| \medskip\\
& \quad\quad\overset{(b)}{\ge} E_{i^k_\rho}(\x^k) \overset{(c)}{\ge} \rho\, E_{i^k_{\text{mx}}}(\x^k)\medskip\\
&\quad\quad\overset{(d)}{\ge}  {\rho}  \,\left({\min_{i\in \mathcal{N}}\underline{s}_{i}}\right)\, \left(\max_{i\in \mathcal{N}}\left\|\hat{\x}_i(\x^k)-\x_i^k\right\|\right)\medskip\\
&\quad\quad\,\,\,{\ge}\, \frac{\rho}{N} \,\left({\min_{i\in \mathcal{N}}\underline{s}_{i}}\right)\, \left\|\hat{\x}(\x^k)-\x^k\right\|
\end{aligned}
\]
where: in (a)  $i^k_{\rho}$  is any index in  $\hat{\mathcal{S}}^k$  such that  $E_{i^k_{\rho}}(\x^k)\geq \rho\, \max_{i\in \mathcal{S}^k} E_i(\x^k)$.   Note that by definition of $\hat{\mathcal{S}}^k$ (cf. step S.3 of Algorithm 1), such a index  always exists;   (b) is due to \eqref{eq:error bound};    (c) follows from the definition of ${i^k_\rho}$, and $\max_{i\in \mathcal{S}^k} E_i(\x^k)= E_{i^k_{\text{mx}}}(\x^k)$, the latter due to $i^k_{\text{mx}}\in \mathcal{S}^k\supseteq \hat{\mathcal{S}}^k$ (recall that $k\in \Km$); and (d) follows from \eqref{eq:error bound}.
\end{proof} \vspace{-0.3cm}


\begin{lemma}\label{lemma on errors} Let $\{\x^k\}$ be the sequence generated by Algorithm 1. For every  $k\in \mathbb{N}_+$, and  $\hat{\mathcal{S}}^k$ generated as  in step S.3,   the following holds:
\begin{equation}\label{error_descent}
\begin{array}{ll}
\left( \nabla_{\x}F (\x^{k})\right)^{T}_{\tiny { \hat{\mathcal{S}}^k}} \left(\widehat{\x}(\x^k) -\x^k \right)_{\tiny { \hat{\mathcal{S}}^k}}\leq - q \,\| \left(\widehat{\x}(\x^k)-\x^k\right)_{\tiny { \hat{\mathcal{S}}^k}}\|^2 \medskip\\\hspace{3cm}    +\displaystyle{\sum_{i \in {\tiny { \hat{\mathcal{S}}^k}}}} \left[G(\x^k)-G(\widehat{\x}_i(\x^k),\x_{-i}^k)\right].
\end{array} \vspace{-0.3cm}
\end{equation} 
\end{lemma} 
\begin{proof}
Optimality of $\widehat{\x}_i(\x^k)$ for the subproblem $i$ implies
\begin{equation}\label{eq:optsubprob}
\left(\nabla_{\x_i} \widetilde{F}_i(\widehat{\x}_i(\x^k);\x^k)+\boldsymbol{\xi}_i(\widehat{\x}_i(\x^k),\x_{-i}^k)\right)^T\left(\mathbf{y}_i-\widehat{\x}_i(\x^k)\right)\ge 0,\nonumber
\end{equation}
for all $\y_i\in X_i$, and some $\boldsymbol{\xi}_i(\widehat{\x}_i(\x^k),\x_{-i}^k)\in \partial_{\x_i} G (\widehat{\x}_i(\x^k),\x_{-i}^k)$. Therefore,
\begin{equation}\label{eq:optsubprob2}
\begin{array}{ll}
0\ge  \nabla_{\x_i} \widetilde{F}_i(\widehat{\x}_i(\x^k);\x^k)^T\,\left(\widehat{\x}_i(\x^k)-\x_i^k\right)\medskip\\
\hspace{2cm} +\,\boldsymbol{\xi}_i(\widehat{\x}_i(\x^k),\x_{-i}^k)^T\,\left(\widehat{\x}_i(\x^k)-\x_i^k\right).
\end{array}
\end{equation}

Let us (lower) bound next the two terms on the RHS of \eqref{eq:optsubprob2}. The uniform strong monotonicity of $\widetilde{F}_i(\bullet;\x^k)$ (cf. F1),\vspace{-0.1cm}
\begin{equation}\label{eq:monotone}
\begin{array}{ll}
\left(\nabla_{\x_i} \widetilde{F}_i(\widehat{\x}_i(\x^k);\x^k)-\nabla_{\x_i} \widetilde{F}_i(\x_i^k;\x^k)\right)^T (\widehat{\x}_i(\x^k)-\x^k_i)\medskip\\\hspace{2cm}\ge \,q \,||\widehat{\x}_i(\x^k)-\x^k_i||^2,
\end{array}
\end{equation}
along with  the gradient consistency condition (cf. F2) $\nabla_{\x_i} \widetilde{F}_i(\x_i^k;\x^k) = \nabla_{\x_i} F(\x^k)$ imply\\\vspace{-0.2cm}
\begin{equation}\label{eq:monotone2}
\begin{array}{l}
\nabla_{\x_i} \widetilde{F}_i(\widehat{\x}_i(\x^k);\x^k)^T\,\left(\widehat{\x}_i(\x^k)-\x_i^k\right)\medskip\\
=\left(\nabla_{\x_i} \widetilde{F}_i(\widehat{\x}_i(\x^k);\x^k)-\nabla_{\x_i} \widetilde{F}_i(\x_i^k;\x^k)\right)^T
\left(\widehat{\x}_i(\x^k)-\x^k_i\right) \medskip\\ \quad\,\,+  \nabla_{\x_i} \widetilde{F}_i(\x_i^k;\x^k)^T\left(\widehat{\x}_i(\x^k)-\x^k_i\right)\medskip\\  \ge \nabla_{\x_i} F(\x^k)^T\left(\widehat{\x}_i(\x^k)-\x^k_i\right) + q\,||\widehat{\x}_i(\x^k)-\x^k_i||^2.
\end{array}
\end{equation}
To bound the second term on the RHS of \eqref{eq:optsubprob2}, let us invoke the convexity of $G(\bullet,\x^k_{-i})$:
\begin{equation}
\label{eq:convexity}\nonumber
\begin{array}{ll}
G(\x^k_i,\x^k_{-i})-G(\widehat{\x}_i(\x^k),\x^k_{-i})\medskip\\\qquad\qquad \ge \boldsymbol{\xi}_i(\widehat{\x}_i(\x^k),\x_{-i}^k)^T\, \left(\x^k_i-\widehat{\x}_i(\x^k)\right),
\end{array}
\end{equation}
which yields
\begin{equation}
\label{eq:convexity2}
\begin{array}{ll}
\boldsymbol{\xi}_i(\widehat{\x}_i(\x^k),\x_{-i}^k)^T\, \left(\widehat{\x}_i(\x^k)-\x^k_i\right)\medskip\\\qquad\qquad
\geq G(\widehat{\x}_i(\x^k),\x^k_{-i})-G(\x^k).
\end{array}
\end{equation}

The desired result (\ref{error_descent}) is readily obtained by combining~\eqref{eq:optsubprob2} with~\eqref{eq:monotone2} and~\eqref{eq:convexity2}, and summing over $i\in \hat{\mathcal{S}}^k$.
\end{proof}

\begin{lemma}\label{bounds_on_G} Let $\{\x^k\}$ be the sequence generated by Algorithm 1,   and $\{\gamma^k\}\!\downarrow 0$. For every  $k\in \mathbb{N}_+$ sufficiently large, and $\hat{\mathcal{S}}^k$ generated as in step S.3, the following holds:
\begin{equation}\label{G_k_bound}
\begin{array}{ll}
G(\x^{k+1})\leq G(\x^k)+ \gamma^k\,L_G \sum_{i \in {\hat{\mathcal{S}}^k}}\varepsilon_{i}^{k}\medskip\\\hspace{2.8cm}    + {\gamma}^k\, \displaystyle{\sum_{i\in \hat{\mathcal{S}}^k}} \left[G(\widehat{\x}_i(\x^k),\x^k_{-i})-G({\x}^k)\right].
\end{array}
\end{equation} 
\end{lemma}\vspace{-0.4cm}
\begin{proof}
 Given $k\geq 0$ and $\hat{\mathcal{S}}^{k}$, define
  $\bar{\x}^{k}\triangleq (\bar{\x}_i^k)_{i\in \mathcal{N}}$, with \begin{equation} \bar{\mathbf{x}}_{i}^{k}\triangleq\left\{ \begin{array}{ll}
\x_i^k+\gamma^k\,\left(\widehat{\mathbf{x}}_{i}(\mathbf{x}^{k})-\x^k_i\right),\ensuremath{} & \mbox{if }i\in\hat{\mathcal{S}}^{k}\\
\mathbf{x}_{i}^{k} & \mbox{otherwise}.
\end{array}\right.\nonumber\end{equation}

 By the convexity and Lipschitz continuity of $G$, it follows
\begin{equation}\label{eq:pinexactbreak}
\begin{array}{ll}
G(\x^{k+1}) &=\,\, G(\x^k)+\left(G(\x^{k+1})-G(\bar{\x}^k)\right) \medskip \\ &\quad \,\,+ \left(G(\bar{\x}^k)-G(\x^k)\right)\medskip \\ & \le
\,\,G(\x^k)+\gamma^k\, L_G\,\sum_{i\in\hat{\mathcal{S}}^k} \varepsilon^k_i\medskip\\ & \quad\,\, + \left(G(\bar{\x}^k)-G(\x^k)\right),
\end{array}
\end{equation}
 where $L_G$ is a (global) Lipschitz constant of $G$. We bound next the last term on the RHS of \eqref{eq:pinexactbreak}.

 Let $\bar{\gamma}^k=\gamma^k N$, for $k$ large enough so that $0<\bar{\gamma}^k<1$. Define $\check{\x}^{k}\triangleq (\check{\x}_i^k)_{i\in \mathcal{N}}$, with $\check{\x}_i^k=\x_i^k$ if $i\notin \hat{\mathcal{S}}^k$, and 
\begin{equation}\label{eq:def_x_i_check}
\check{\x}^{k}_i\triangleq \bar{\gamma}^k\, \widehat{\x}_i(\x^k)+(1-\bar{\gamma}^k)\,\x^k_i
\end{equation}
otherwise. Using the definition of $\bar{\x}^k$ it is not difficult to see  that\vspace{-0.2cm}
\begin{equation}\label{x_bar}
\bar{\x}^k= \frac{N-1}{N}\,\x^k+\frac{1}{N}\, \check{\x}^k.
\end{equation}

Using (\ref{x_bar}) and invoking the convexity of $G$, the following recursion holds for sufficiently large $k$:
\begin{equation}
\begin{array}{l}
G(\bar{\x}^k) = G\left(\frac{1}{N}(\check{\x}^k_1,\x^k_{-1})+\frac{1}{N}(\x^k_1,\check{\x}^k_{-1}) +
\frac{N-2}{N}\x^k\right)\medskip \\ \quad =G\left(\frac{1}{N}\,(\check{\x}^k_1,\x^k_{-1}) +\frac{N-1}{N}\,\left(\x^k_1,\frac{1}{N-1}\,\check{\x}^k_{-1}+\frac{N-2}{N-1}\,\x^k_{-1}\right)\right)\medskip \\
\quad \leq \frac{1}{N}\,G\left(\check{\x}^k_1,\x^k_{-1}\right) +\frac{N-1}{N}\,G\left(\x^k_1,\frac{1}{N-1}\,\check{\x}^k_{-1}+\frac{N-2}{N-1}\,\x^k_{-1}\right)\medskip\\
\quad = \frac{1}{N}\,G\left(\check{\x}^k_1,\x^k_{-1}\right)+\frac{N-1}{N}\,G\left(\frac{1}{N-1}\left(\x^k_1,\check{\x}^k_{-1}\right)+\frac{N-2}{N-1}\x^k\right)
\end{array}\nonumber
\end{equation}
\begin{equation}
\label{eq:pconvexxbar}
\begin{array}{l}
\quad =
\frac{1}{N} G\left(\check{\x}^k_1,\x^k_{-1}\right) +\frac{N-1}{N} \,G\left(\frac{1}{N-1}\,
\left(\check{\x}^k_2,\x^k_{-2}\right)\right.\medskip\\
\hspace{3.3cm}
\left.+\frac{1}{N-1}\left(\x^k_1,\x^k_2,\check{\x}^k_{-(1,2)}\right)
  + \frac{N-3}{N-1}\,\x^k\right)\medskip \\
  
  \quad =\frac{1}{N} G\left(\check{\x}^k_1,\x^k_{-1}\right) +\frac{N-1}{N} \,G\left(\frac{1}{N-1}\,
\left(\check{\x}^k_2,\x^k_{-2}\right)\right.\medskip\\
\hspace{1.8cm}
\left.+\frac{N-2}{N-1}\left(\x^k_1,\x^k_2,\frac{1}{N-2}\,\check{\x}^k_{-(1,2)}+\frac{N-3}{N-2}\,\x^k_{-(1,2)}\right)\right)\medskip\\
\quad \le \frac{1}{N} G\left(\check{\x}^k_1,\x^k_{-1}\right) + \frac{1}{N}\,G\left(\check{\x}^k_2,\x^k_{-2}\right)\medskip \\
\quad  \quad + \frac{N-2}{N-1}\,G\left(\x^k_1,\x^k_2,\frac{1}{N-2}\,\check{\x}^k_{-(1,2)}+\frac{N-3}{N-2}\,\x^k_{-(1,2)}\right) \medskip \\
\quad \le \quad ... \quad \le \dfrac{1}{N}\, \displaystyle{\sum_{i\in \mathcal{N}}} G(\check{\x}^k_i,\x^k_{-i}).
\end{array} 
\end{equation}

Using \eqref{eq:pconvexxbar}, the last term on the RHS of \eqref{eq:pinexactbreak} can be upper bounded for $k$ sufficiently large as
\begin{equation}\label{bound_G_bar_G_k}
\begin{array}{ll}
&\!\!\!\!\!\!\!G(\bar{\x}^k)-G({\x}^k) \leq \dfrac{1}{N}\,\displaystyle{\sum_{i\in \mathcal N}} \left[G(\check{\x}^k_i,\x^k_{-i})-G({\x}^k)\right]\medskip\\
&= \dfrac{1}{N}\,\displaystyle{\sum_{i\in \hat{\mathcal{S}}^k}} \left[G(\check{\x}^k_i,\x^k_{-i})-G({\x}^k)\right]\medskip\\
&\overset{(a)}{\leq} \dfrac{1}{N} \displaystyle{\sum_{i\in \hat{\mathcal{S}}^k}} \left[\bar{\gamma}^k G(\widehat{\x}_i(\x^k),\x^k_{-i})+ (1-\bar{\gamma}^k) G(\x^k) -G({\x}^k)\right]\medskip\\
&= {\gamma}^k\, \displaystyle{\sum_{i\in \hat{\mathcal{S}}^k}} \left[G(\widehat{\x}_i(\x^k),\x^k_{-i})-G({\x}^k)\right],
\end{array}
\end{equation}
  where (a) is due to the convexity  of $G(\bullet,\x^k_{-i})$ and the definition of $\check{\x}^k_i$ [cf. (\ref{eq:def_x_i_check})].

  The desired inequality (\ref{G_k_bound}) follows readily by combining (\ref{eq:pinexactbreak}) with  (\ref{bound_G_bar_G_k}).
\end{proof}

\begin{lemma} \emph{\cite[Lemma 3.4, p.121]{Bertsekas-Tsitsiklis_bookNeuro11}}\label{lemma_Robbinson_Siegmunt}
Let $\{X^{k}\}$, $\{Y^{k}\}$, and $\{Z^{k}\}$ be three sequences
of numbers such that $Y^{k}\geq0$ for all $k$. Suppose that
\[
X^{k+1}\leq X^{k}-Y^{k}+Z^{k},\quad\forall k=0,1,\ldots
\]
and $\sum_{k=0}^\infty Z^{k}<\infty$. Then either $X^{k}\rightarrow-\infty$
or else $\{X^{k}\}$ converges to a finite value and $\sum_{k=0}^\infty Y^{k}<\infty$.
\end{lemma}

\subsection{Proof of Theorem \ref{Theorem_convergence_inexact_Jacobi}}\label{proof_Th1} 
%

For any given $k\geq 0$, the Descent Lemma \cite{Bertsekas_Book-Parallel-Comp}
yields: with $\widehat{\mathbf{z}}^{k}\triangleq(\widehat{\mathbf{z}}_{i}^{k})_{i\in \mathcal{N}}$ and
$\mathbf{z}^{k}\triangleq(\mathbf{z}_{i}^{k})_{i\in \mathcal{N}}$ defined in step S.4  of Algorithm \ref{alg:general},
{
\begin{equation}
\begin{array}{lll}
  F\left(\mathbf{x}^{k+1}\right)  & \leq &  F\left(\mathbf{x}^{k}\right)+\gamma^{k}\,\nabla_{\mathbf{x}}F\left(\mathbf{x}^{k}\right)^{T}\left(\widehat{\mathbf{z}}^{k}-\mathbf{x}^{k}\right)\smallskip\\
  && +\dfrac{\left(\gamma^{k}\right)^{2}{L_{\nabla F}}}{2}\,\left\Vert \widehat{\mathbf{z}}^{k}-\mathbf{x}^{k}\right\Vert ^{2}.
\end{array}\label{eq:descent_Lemma}
\end{equation}
 We bound next the second and third terms on the RHS of \eqref{eq:descent_Lemma}.
Denoting by $\overline{\hat{\mathcal{S}}}^k$ the complement of $\hat{\mathcal{S}}^k$, we   have, 
\begin{equation}\label{nabla_times_z_hat}
\begin{array}{l}
\nabla_{\mathbf{x}}F\left(\mathbf{x}^{k}\right)^{T}\left(\widehat{\mathbf{z}}^{k}-\mathbf{x}^{k}\right)\smallskip\\
\qquad\qquad   =
\nabla_{\mathbf{x}}F\left(\mathbf{x}^{k}\right)^{T}\left(\widehat{\mathbf{z}}^{k}-\widehat{\mathbf{x}}(\mathbf{x}^{k})
+\widehat{\mathbf{x}}(\mathbf{x}^{k})-\mathbf{x}^{k}\right)\smallskip\\
\qquad\qquad \overset{(a)}{=}  \nabla_{\mathbf{x}}F\left(\mathbf{x}^{k}\right)^{T}_{\hat{\mathcal{S}}^k} (\mathbf{z}^k - \widehat{\mathbf{x}}(\mathbf{x}^k))_{\hat{\mathcal{S}}^k}\smallskip\\ \qquad\qquad\quad
 +
 \nabla_{\mathbf{x}}F\left(\mathbf{x}^{k}\right)^{T}_{\overline{\hat{\mathcal{S}}}^k} (\mathbf{x}^k - \widehat{\mathbf{x}}(\mathbf{x}^k))_{\overline{\hat{\mathcal{S}}}^k}\smallskip\\
\qquad \qquad \quad   +  \nabla_{\mathbf{x}}F\left(\mathbf{x}^{k}\right)^{T}_{\hat{\mathcal{S}}^k} (\widehat{\mathbf{x}}(\mathbf{x}^k)-\mathbf{x}^k)_{\hat{\mathcal{S}}^k}\smallskip\\ \qquad\qquad\quad
 +    \nabla_{\mathbf{x}}F\left(\mathbf{x}^{k}\right)^{T}_{\overline{\hat{\mathcal{S}}}^k} (\widehat{\mathbf{x}}(\mathbf{x}^k)-\mathbf{x}^k)_{\overline{\hat{\mathcal{S}}}^k}\smallskip\\
\qquad\qquad   = \nabla_{\mathbf{x}}F\left(\mathbf{x}^{k}\right)^{T}_{\hat{\mathcal{S}}^k} (\mathbf{z}^k - \widehat{\mathbf{x}}(\mathbf{x}^k))_{\hat{\mathcal{S}}^k} \smallskip\\ \qquad\qquad\quad +
 \nabla_{\mathbf{x}}F\left(\mathbf{x}^{k}\right)^{T}_{\hat{\mathcal{S}}^k} (\widehat{\mathbf{x}}(\mathbf{x}^k)-\mathbf{x}^k)_{\hat{\mathcal{S}}^k}\medskip\\
 \qquad\qquad   \overset{(b)}{\leq} \displaystyle{\sum_{i \in \hat{\mathcal{S}}^k}}\varepsilon_{i}^{k}\left\Vert \nabla_{\mathbf{x}_{i}}F(\mathbf{x}^{k})\right\Vert\medskip\\
 \ \qquad\qquad\quad +\nabla_{\mathbf{x}}F\left(\mathbf{x}^{k}\right)^{T}_{\hat{\mathcal{S}}^k} (\widehat{\mathbf{x}}(\mathbf{x}^k)-\mathbf{x}^k)_{\hat{\mathcal{S}}^k}\medskip\\
 \qquad\qquad  \overset{(c)}{\leq}   \displaystyle{\sum_{i \in \hat{\mathcal{S}}^k}}\varepsilon_{i}^{k}\left\Vert \nabla_{\mathbf{x}_{i}}F(\mathbf{x}^{k})\right\Vert\medskip\\\qquad\qquad\quad- q \,\| \left(\widehat{\x}(\x^k)-\x^k\right)_{\tiny { \hat{\mathcal{S}}^k}}\|^2  \medskip\\\qquad\qquad\quad    +\displaystyle{\sum_{i \in {\tiny { \hat{\mathcal{S}}^k}}}} \left[G(\x^k)-G(\widehat{\x}_i(\x^k),\x_{-i}^k)\right]
\end{array}
\end{equation}
where in (a) we used the definition of $\widehat{\mathbf{z}}^k$ and of the set $\hat{\mathcal{S}}^k$;  in (b)  we used $\left\Vert \mathbf{z}_{i}^{k}-\widehat{\mathbf{x}}_{i}(\mathbf{x}^{k})\right\Vert \leq\varepsilon_{i}^{k}$; and (c) follows from (\ref{error_descent}) (cf. Lemma \ref{lemma on errors}).

The third term on the RHS of \eqref{eq:descent_Lemma} can be bounded as
\begin{equation}\label{eq:49bis} \begin{array}{rcl}
 \left\Vert \widehat{\mathbf{z}}^{k}-\mathbf{x}^{k}\right\Vert ^{2}&\leq& 2\,
 \left\Vert \left(\mathbf{z}^{k}-\hat{\mathbf{x}}(\x^{k})\right)_{\hat{\mathcal{S}}^k}\right\Vert ^{2}\smallskip\\
 & &+   2\, \left\Vert \left(\hat{\mathbf{x}}(\x^{k})-\x^k\right)_{\hat{\mathcal{S}}^k}\right\Vert ^{2}\medskip\\
 & = &
 +2\sum_{i\in  {\hat{\mathcal{S}}^k}}\left\Vert \mathbf{z}_{i}^{k}-\widehat{\mathbf{x}}_{i}(\mathbf{x}^{k})\right\Vert ^{2}\smallskip
\\ && + 2\, \left\Vert \left(\hat{\mathbf{x}}(\x^{k})-\x^k\right)_{\hat{\mathcal{S}}^k}\right\Vert ^{2}\medskip\\
& \leq&  2\,\displaystyle{\sum_{i \in {\hat{\mathcal{S}}^k}}}(\varepsilon_{i}^{k})^{2} + 2\left\Vert \left(\hat{\mathbf{x}}(\x^{k})-\x^k\right)_{\hat{\mathcal{S}}^k}\right\Vert ^{2},
 \end{array}
 \end{equation}
 where the first inequality follows from the definition of $ \mathbf{z}^{k}$ and $ \widehat{\mathbf{z}}^{k}$, and
 in the last inequality we used $\left\Vert \mathbf{z}_{i}^{k}-\widehat{\mathbf{x}}_{i}(\mathbf{x}^{k})\right\Vert \leq\varepsilon_{i}^{k}$.

Now, we combine the above results to get the descent property of $V$ along $\{\x^k\}$. For sufficiently large $k\in \mathbb{N}_+$, it  holds
\begin{equation}
\begin{array}{l}
\hspace{-0.4cm}V(\mathbf{x}^{k+1})  =  F(\mathbf{x}^{k+1}) + G(\mathbf{x}^{k+1}) \\[0.3em]
{\leq}
V\left(\mathbf{x}^{k}\right)-\gamma^{k}\left(q -\gamma^{k}{L_{\nabla F}}\right)\left\Vert \left(\widehat{\mathbf{x}}(\mathbf{x}^{k})-\mathbf{x}^{k}\right)_{\hat{\mathcal S}^k}\right\Vert ^{2}+T^{k},
\end{array}\label{eq:descent_Lemma_2}
\end{equation}
where the inequality follows from (\ref{error_descent}), (\ref{eq:descent_Lemma}), (\ref{nabla_times_z_hat}),  and (\ref{eq:49bis}),   and  $T^{k}$ is given by
$$T^{k}\triangleq\gamma^{k}\,
\sum_{i \in\cal N}\varepsilon_{i}^{k}\left( L_G +
\left\Vert \nabla_{\mathbf{x}_{i}}F(\mathbf{x}^{k})\right\Vert\right) +\left(\gamma^{k}\right)^{2}{L_{\nabla F}}\,\sum_{i\in \cal N}(\varepsilon_{i}^{k})^{2}.$$
By assumption (iv) in Theorem \ref{Theorem_convergence_inexact_Jacobi}, it is not difficult to show that
 $\sum_{k=0}^{\infty}T^{k}<\infty$.
Since $\gamma^{k}\rightarrow 0$, it follows from (\ref{eq:descent_Lemma_2}) that there exist  some positive constant
$\beta_{1}$ and a sufficiently large $k$, say $\bar{{k}}$, such that
\begin{equation}
V(\mathbf{x}^{k+1})\leq V(\mathbf{x}^{k})-\gamma^{k}\beta_{1}\left\Vert \left(\widehat{\mathbf{x}}(\mathbf{x}^{k})-\mathbf{x}^{k}\right)_{\hat{\mathcal S}^k}\right\Vert ^{2}+T^{k},\label{eq:descent_Lemma_3_}
\end{equation}
for all $k\geq\bar{{k}}$. Invoking Lemma \ref{lemma_Robbinson_Siegmunt} 
while using $\sum_{k=0}^\infty T^{k}<\infty$ and the coercivity of $V$, we deduce
from (\ref{eq:descent_Lemma_3_}) that
\begin{equation}
\begin{array}{ll}
\displaystyle{\lim_{t\rightarrow\infty}}\,\sum_{k=\bar{{k}}}^{t}\gamma^{k}\left\Vert \left(\widehat{\mathbf{x}}(\mathbf{x}^{k})-\mathbf{x}^{k}\right)_{\hat{\mathcal S}^k}\right\Vert ^{2}<+\infty,
\end{array}
\end{equation}
and thus also
\begin{equation}
\begin{array}{ll}
\displaystyle{\lim_{t\rightarrow\infty}}\,\sum_{{\mathcal{K}_{\text{mx}}}\ni \,k\,\geq\, \bar{{k}}}^{t}\gamma^{k}\left\Vert \left(\widehat{\mathbf{x}}(\mathbf{x}^{k})-\mathbf{x}^{k}\right)_{\hat{\mathcal S}^k}\right\Vert ^{2}<+\infty.
\end{array}\label{eq:finite_sum_series}
\end{equation}

 Lemma \ref{lemma:almost_surely_nonsum}  together with~\eqref{eq:finite_sum_series}   imply
\[
\liminf_{k\in\mathcal{K}_{\text{mx}}} \left\Vert \left(\widehat{\mathbf{x}}(\mathbf{x}^{k})-\mathbf{x}^{k}\right)_{\hat{\mathcal S}^k}\right\Vert  = 0, \qquad \text{w.p.\,} 1,
\]
which by Lemma~\ref{lem:sampledesc} implies
\begin{equation}\label{lim_inf_convergence}
\liminf_{k\to\infty} \left\Vert \widehat{\mathbf{x}}(\mathbf{x}^{k})-\mathbf{x}^{k} \right\Vert = 0,\qquad \text{w.p.\,} 1.
\end{equation}
 Therefore,  the limit point of the infimum sequence is a fixed point of $\widehat{\x}(\cdot)$ w.p.1.

\subsection{Proof of Theorem 3}\label{proof_Th2}

The proof follows similar ideas as the one of Theorem  1 in our recent work \cite{FacSagScuTSP14}, but with the nontrivial complication of dealing with randomness in the block selection.

Given \eqref{lim_inf_convergence},  we show next that, under the separability assumption on $G$, it holds that $\lim_{k\rightarrow\infty}\left\Vert \widehat{\mathbf{x}}(\mathbf{x}^{k})-\mathbf{x}^{k}\right\Vert =0$ w.p.1. For notational  simplicity, let us define  $\triangle\widehat{\mathbf{x}}(\mathbf{x}^{k})\triangleq\widehat{\mathbf{x}}(\mathbf{x}^{k})-\mathbf{x}^{k}$.

Note first  that for any finite but arbitrary sequence  $\{k, k+1, ..., i_k-1\}$, it holds that
\[
\mathbb{E}\left[\sum_{{{\mathcal K}_{\text{mx}}}\ni t=k}^{i_k-1}\gamma^t\right] =
\sum_{t=k}^{i_k-1}\gamma^t \left[\mathbb{P}(t\in {\mathcal K}_{\text{mx}})\right] \ge
p \,\sum_{t=k}^{i_k-1}\gamma^t,
\]
 and thus
\[
\mathbb{P}\left(\sum_{{{\mathcal K}_{\text{mx}}}\ni t=k}^{i_k-1}\gamma^t > {\beta} \sum_{t=k}^{i_k-1}\gamma^t\right) > 0,
\]
for all $ k\in\mathcal{K}$ and $0<\beta<p$.
This implies that, w.p.1, there exists  an infinite sequence of indexes, say $\mathcal{K}_1\subseteq \mathcal{K}$, such that
\begin{equation}\label{lower_bound_sum_gamma}
\sum_{{{\mathcal K}_{\text{mx}}}\ni t=k}^{i_k-1}\gamma^t > {\beta} \sum_{t=k}^{i_k-1}\gamma^t,\quad \forall k\in \mathcal{K}_1.
\end{equation}
 Suppose now, by contradiction, that $\limsup_{k\rightarrow\infty}\left\Vert \triangle\widehat{\mathbf{x}}(\mathbf{x}^{k})\right\Vert >0$ with a positive probability.
 Then we can find a realization such that at the same time \eqref{lower_bound_sum_gamma} holds for some $\mathcal{K}_1$ and $\limsup_{k\rightarrow\infty}\left\Vert \triangle\widehat{\mathbf{x}}(\mathbf{x}^{k})\right\Vert >0$. In the rest of the proof we focus on this realization and get a contradiction, thus proving
 that $\limsup_{k\rightarrow\infty}\left\Vert \triangle\widehat{\mathbf{x}}(\mathbf{x}^{k})\right\Vert =0$ w.p.1.

If $\limsup_{k\rightarrow\infty}\left\Vert \triangle\widehat{\mathbf{x}}(\mathbf{x}^{k})\right\Vert >0$
then there exists a $\delta>0$ such that $\left\Vert \triangle\widehat{\mathbf{x}}(\mathbf{x}^{k})\right\Vert >2\delta$
for infinitely many $k$ and also $\left\Vert \triangle\widehat{\mathbf{x}}(\mathbf{x}^{k})\right\Vert <\delta$
for infinitely many $k$. Therefore, one can always find an infinite
set of indexes, say $\mathcal{K}$, having the following properties:
for any $k\in\mathcal{K}$, there exists an integer $i_{k}>k$ such
that
\begin{eqnarray}
\left\Vert \triangle\widehat{\mathbf{x}}(\mathbf{x}^{k})\right\Vert <\delta, &  & \left\Vert \triangle\widehat{\mathbf{x}}(\mathbf{x}^{i_{k}})\right\Vert >2\delta\medskip\label{eq:outside_interval}\\
\delta\leq\left\Vert \triangle\widehat{\mathbf{x}}(\mathbf{x}^{j})\right\Vert \leq2\delta &  & k<j<i_{k}.\label{eq:inside_interval}
\end{eqnarray}


 {Proceeding now as in  the proof of Theorem \ref{Theorem_convergence_inexact_Jacobi} in~\cite{FacSagScuTSP14}, we have: for  $k\in\mathcal{K}_1$,
\begin{eqnarray}
\delta & \overset{(a)}{<} & \left\Vert \triangle\widehat{\mathbf{x}}(\mathbf{x}^{i_{k}})\right\Vert -\left\Vert \triangle\widehat{\mathbf{x}}(\mathbf{x}^{k})\right\Vert \medskip\nonumber \\
 & \leq & \left\Vert \widehat{\mathbf{x}}(\mathbf{x}^{i_{k}})-\widehat{\mathbf{x}}(\mathbf{x}^{k})\right\Vert +\left\Vert \mathbf{x}^{i_{k}}-\mathbf{x}^{k}\right\Vert \\
 & \overset{(b)}{\leq} & (1+\hat{{L}})\left\Vert \mathbf{x}^{i_{k}}-\mathbf{x}^{k}\right\Vert \\
 & \overset{(c)}{\leq} &      (1+\hat{{L}})\sum_{t=k}^{i_{k}-1}\gamma^{t}\left(\left\Vert \triangle\widehat{\mathbf{x}}(\mathbf{x}^{t})_{S^t}\right\Vert +\left\Vert (\mathbf{z}^{t}-\widehat{\mathbf{x}}(\mathbf{x}^{t}))_{S^t}\right\Vert \right)\vspace{-0.3cm}\nonumber \\
 & \overset{(d)}{\leq} & (1+\hat{{L}})\,(2\delta+\varepsilon^{\max})\sum_{t=k}^{i_{k}-1}\gamma^{t},\label{eq:lower_bound_sum}
\end{eqnarray}
where (a) follows from (\ref{eq:outside_interval});
(b) is due to Lemma \ref{Prop_x_y}; (c) comes from the triangle
inequality, the updating rule of the algorithm   and the definition of $\widehat{\z}^k$; and in (d) we used
(\ref{eq:outside_interval}), (\ref{eq:inside_interval}), and $\left\Vert \mathbf{z}^{t}-\widehat{\mathbf{x}}(\mathbf{x}^{t})\right\Vert \leq\sum_{i\in \cal N}\varepsilon_{i}^{t}$,
where $\varepsilon^{\max}\triangleq\max_{k}\sum_{i\in\cal N}\varepsilon_{i}^{k}<\infty$.
It follows from (\ref{eq:lower_bound_sum}) that
\begin{equation}
\liminf_{{\mathcal{K}_1}\ni k\rightarrow\infty}\sum_{t=k}^{i_{k}-1}\gamma^{t}\geq\dfrac{{\delta}}{(1+\hat{{L}})(2\delta+\varepsilon^{\max})}>0.\label{eq:lim_inf_bound}
\end{equation}

{We show next that (\ref{eq:lim_inf_bound}) is in contradiction with
the convergence of $\{V(\mathbf{x}^{k})\}$. To do that, we preliminary
prove that, for sufficiently large $k\in\mathcal{K}$, it must be
$\left\Vert \triangle\widehat{\mathbf{x}}(\mathbf{x}^{k})\right\Vert \geq\delta/2$.
Proceeding as in (\ref{eq:lower_bound_sum}), we have: for any given
$k\in\mathcal{K}$,
\[
\begin{array}{l}
\left\Vert \triangle\widehat{\mathbf{x}}(\mathbf{x}^{k+1})\right\Vert -\left\Vert \triangle\widehat{\mathbf{x}}(\mathbf{x}^{k})\right\Vert \leq(1+\hat{{L}})\left\Vert \mathbf{x}^{k+1}-\mathbf{x}^{k}\right\Vert \smallskip\\
\qquad\qquad\qquad\qquad\qquad\leq(1+\hat{{L}})\gamma^{k}\left(\left\Vert \triangle\widehat{\mathbf{x}}(\mathbf{x}^{k})\right\Vert +\varepsilon^{\max}\right).
\end{array}
\]
 It turns out that for sufficiently large $k\in\mathcal{K}_1$ so that
$(1+\hat{{L}})\gamma^{k}<\delta/(\delta+2\varepsilon^{\max})$, it
must be
\begin{equation}
\left\Vert \triangle\widehat{\mathbf{x}}(\mathbf{x}^{k})\right\Vert \geq\delta/2;\label{eq:lower_bound_delta_x_n}
\end{equation}
otherwise the condition $\left\Vert \triangle\widehat{\mathbf{x}}(\mathbf{x}^{k+1})\right\Vert \geq\delta$
would be violated [cf. (\ref{eq:inside_interval})]. Hereafter
we assume without loss of generality  that (\ref{eq:lower_bound_delta_x_n}) holds for
all $k\in\mathcal{K}_1$ (in fact, one can always restrict $\{\mathbf{x}^{k}\}_{k\in\mathcal{K}_1}$
to a proper subsequence).
}

We can show now that (\ref{eq:lim_inf_bound}) is in contradiction
with the convergence of $\{V(\mathbf{x}^{k})\}$. Using (\ref{eq:descent_Lemma_3_})
(possibly over a subsequence), we have: for sufficiently large $k\in\mathcal{K}_1$,
\begin{equation}\label{eq:gammat}
\begin{array}{l}
V(\x^{i_k}) \le V(\x^k)-\beta_1 \displaystyle{\sum_{{\mathcal{K}_\text{mx}}\ni t=k}^{i_k-1}}\!\! \gamma^t \left\Vert \left(\triangle\widehat{\mathbf{x}}(\mathbf{x}^{t})\right)_{\hat{\mathcal{S}}^t}\right\Vert^2
+\!\!\displaystyle{\sum_{{\mathcal{K}_\text{mx}}\ni t=k}^{i_k-1}}  T^t \medskip
\\ \qquad \overset{(a)}{\le} V(\x^k)- \beta_2 \displaystyle{\sum_{{\mathcal{K}_\text{mx}}\ni t=k}^{i_k-1}}\!\! \gamma^t \left\Vert  \triangle\widehat{\mathbf{x}}(\mathbf{x}^{t}) \right\Vert^2
+\sum_{t=k}^{i_k-1} T^t \medskip
\\ \qquad \overset{(b)}{\le}  V(\x^k)- \beta_3 \,\displaystyle{\sum_{t=k}^{i_k-1}} \gamma^t
+\sum_{t=k}^{i_k-1} T^t,
\end{array}
\end{equation}
where (a) follows from   Lemma~\ref{lem:sampledesc} and $\beta_2=c_1\,\beta_1>0$; and (b) is due to (\ref{eq:lower_bound_delta_x_n}) and (\ref{lower_bound_sum_gamma}), with $\beta_3=\beta\,\beta_2\, (\delta^2/4)$.

Since $\{V(\x^k)\}$ converges and $\sum_{k=0}^{\infty} T^k<\infty$,
it holds that $\lim_{\mathcal{K}_1 \ni k\to\infty}\sum_{t=k}^{i_k-1} \gamma^t=0$, contradicting~\eqref{eq:lim_inf_bound}. Therefore  $\lim_{k\rightarrow\infty}\left\Vert \widehat{\mathbf{x}}(\mathbf{x}^{k})-\mathbf{x}^{k}\right\Vert =0$ w.p.1. Since $\{\x^k\}$ is bounded by the coercivity of $V$ and the convergence of $\{V(\x^k)\}$, it has at least
one limit point $\bar{\x}\in X$. By the continuity of $\widehat{\x}(\bullet)$  (cf. Lemma~\ref{Prop_x_y}) it holds that
$\widehat{\x}(\bar{\x})=\bar{\x}$. By Proposition  \ref{prop_fixed-points} $\bar{{\mathbf{x}}}$ is also a stationary
solution of Problem  (\ref{eq:problem 1}).\hfill$\square$

\bibliographystyle{IEEEtran}
\balance
\bibliography{scutari_facchinei_refs}

\end{document}